\documentclass[10pt,letterpaper]{article}
\usepackage[top=0.85in,left=2.75in,footskip=0.75in]{geometry}

\usepackage{changepage}

\usepackage[utf8]{inputenc}

\usepackage{textcomp,marvosym}

\usepackage{fixltx2e}

\usepackage{amsmath,amssymb}
\usepackage{amsthm}

\usepackage{cite}

\usepackage{nameref,hyperref}

\usepackage{microtype}
\DisableLigatures[f]{encoding = *, family = * }

\usepackage{rotating}

\raggedright
\usepackage[aboveskip=1pt,labelfont=bf,labelsep=period,justification=raggedright,singlelinecheck=off]{caption}

\makeatletter
\renewcommand{\@biblabel}[1]{\quad#1.}
\makeatother

\date{}

\usepackage{lastpage,fancyhdr,graphicx}
\usepackage{epstopdf}

\usepackage[nolist]{acronym}
\usepackage{tabularx}
\usepackage{multirow}
\usepackage{dsfont}
\usepackage{subfig}

\usepackage{lastpage,fancyhdr,graphicx}
\usepackage{epstopdf}

\usepackage[textwidth=5cm]{todonotes}
\usepackage[normalem]{ulem} 

\fancyheadoffset[L]{2.25in}
\fancyfootoffset[L]{2.25in}
\reversemarginpar
\newcommand{\R}{\mathds{R}}
\newcommand{\N}{\mathds{N}}
\newcommand{\UF}{\mathcal{U}}
\newcommand{\HR}{{^*\R}}
\renewcommand{\vec}{\mathbf}
\newcommand{\set}[1]{\left\{#1\right\}}
\newcommand{\supp}{\text{supp}}
\newcommand{\eps}{\varepsilon}
\newcommand{\F}{\mathcal{F}}
\renewcommand{\S}{\mathcal{S}}
\newcommand{\abs}[1]{\left|#1\right|}
\newcommand{\norm}[1]{\left\|#1\right\|}

\theoremstyle{plain}
\newtheorem{thm}{Theorem}[]{\bfseries}{\rmfamily}
\newtheorem{lem}[thm]{Lemma}{\bfseries}{\rmfamily}
{\bfseries}{\rmfamily}
\newtheorem{prop}[thm]{Proposition}{\bfseries}{\rmfamily}
{\bfseries}{\rmfamily}
\newtheorem{defn}[thm]{Definition}{\bfseries}{\rmfamily}
\newtheorem{exa}[thm]{Example}{\bfseries}{\rmfamily}
\newtheorem{rem}[thm]{Remark}{\itshape}{\rmfamily}

\usepackage[symbol]{footmisc}

\begin{document}
\vspace*{0.35in}

\begin{flushleft}
{\Large \textbf\newline{Defending Against Advanced Persistent Threats using
Game-Theory}\footnote[2]{This is a correction to the published version, available from \url{https://journals.plos.org/plosone/article?id=10.1371/journal.pone.0168675}. The correction has been submitted to the journal office in Feb. 2022, and is pending for publication, in which case this version will become replaced.}}
\newline
\\
Stefan Rass\textsuperscript{1,3,*}, %
Sandra K\"{o}nig\textsuperscript{2}, %
Stefan Schauer\textsuperscript{2}, %
\\
\bigskip
\bf{1} Universit\"{a}t Klagenfurt, Department of Artificial Intelligence and Cybersecurity,
Klagenfurt, Austria
\\
\bf{2} Austrian Institute of Technology, Safety \& Security Department,
Klagenfurt, Austria
\\
\bf{3} LIT Secure and Correct Systems Lab, Johannes Kepler University Linz, Austria
\\
\bigskip

%
%





* stefan.rass@aau.at

\end{flushleft}
\section*{Abstract}
Advanced persistent threats (APT) combine a variety of different attack forms
ranging from social engineering to technical exploits. The diversity and
usual stealthiness of APT turns them into a central problem of contemporary
practical system security, since information on attacks, the current system
status or the attacker's incentives is often vague, uncertain and in many
cases even unavailable. Game theory is a natural approach to model the
conflict between the attacker and the defender, and this work investigates a
generalized class of matrix games as a risk mitigation tool for an \ac{APT}
defense. Unlike standard game and decision theory, our model is tailored to
capture and handle the full uncertainty that is immanent to \acp{APT}, such
as disagreement among qualitative expert risk assessments, unknown
adversarial incentives and uncertainty about the current system state (in
terms of how deeply the attacker may have penetrated into the system's
protective shells already). Practically, game-theoretic \ac{APT} models can
be derived straightforwardly from topological vulnerability analysis,
together with risk assessments as they are done in common risk management
standards like the ISO 31000 family. Theoretically, these models come with
different properties than classical game theoretic models, whose technical
solution presented in this work may be of independent interest.

\section{Introduction}
The increasing heterogeneity, connectivity and openness of today's
information systems often lets cyber-attackers find ways into a system on a
considerably large lot of different paths. Today, security is commonly
support by semi-automated tools and techniques to detect and mitigate
vulnerabilities, for example using \ac{TVA}, but this progress is paired with
the parallel evolution and improvements to the related attacks. \acp{APT}
naturally respond to the increasing diversity of security precautions by
mounting attacks in a stealthy and equally diverse fashion, so as to remain
``under the radar'' for as long as is required until the target system has
been penetrated, infected and can be attacked as intended. Countermeasures
may then come too late to be effective any more, since the damage has already
been caused by the time when the attack is detected.

Mitigating \acp{APT} is in most cases not only a matter of technical
precautions, but also some sort of fight against an invisible opponent and
external influences on the system (coming from other connected systems but
primarily due to the \ac{APT} remaining hidden). Thus, any security measure
taken may or may not be effective on the current system state, depending on
how far the \ac{APT} has evolved already. The question of economics then
becomes particularly difficult and fuzzy, since the return on security
investments is almost impossible to quantify in light of many factors that
are outside the security officer's scope of influence.

\subsection{Related Work}
In the last decade, the number of \acp{APT} \cite{Tank11} increased rapidly
and numerous related security incidents were reported all over the world. One
major reason therefore is that \acp{APT} are not focusing on a single
vulnerability in a system (which could be detected and eliminated easily),
but are using a chain of vulnerabilities in different systems to reach
high-security areas within a company network. In this context, adversaries
often exploit the fact that most of the protection efforts go into perimeter
protection, so that moving inside the infrastructure is much easier and the
attacker has a good chance to go unnoticed once being inside. Overcoming the
perimeter protection by social engineering or malware (even unknowingly)
carried inside by legitimate persons (bring-your-own-device problem) are only
two ways to penetrate the perimeter security. Once the perimeter has been
overcome, insider attacks are considered as an even bigger threat
\cite{Cole14}. Extensive guidelines and recommendations exist to secure this
internal area \cite{Sans00}, e.g., the demilitarized zone (DMZ) but the
intensity of the surveillance is limited. Specialized tools for intrusion
detection or intrusion prevention require a large amount of administration
and human resources to monitor the output of these systems.

\acp{APT} are characterized by a combination of several different attack
methods (social engineering, technical hacks, malware, etc.) that is being
tailored to and optimized for the specific organization, its IT network
infrastructure and the existing security measures therein. Often, even yet
not officially reported weaknesses, known as zero-day vulnerabilities, of the
network infrastructure are in additional use. Especially the application of
social engineering in the beginning stages of an \ac{APT} lets the attacker
bypass many technical measures like intrusion detection and prevention
systems, so as to efficiently (and economically) get through the outer
protection (perimeter) of the IT network. A prominent \ac{APT} attack was the
application of the Stuxnet malware in 2008 \cite{FMC11,Karn11,Kush13}, which
was introduced into Iran's nuclear plants sabotaging the nuclear centrifuges.
In the following years, other \ac{APT} attacks, like Operation Aurora, Shady
Rat, Red October or MiniDuke \cite{Tank11,MILP14,FSSF15} have become public.
Additionally, the Mandiant Report \cite{Mand13} explicitly states how
\acp{APT} are used on a global scale for industrial espionage and that the
attackers are often closely connected to governmental organizations.

The detection of APT attacks has therefore become an object of extensive
research over the past years. As perimeter protection tools are occasionally
failing to prevent intrusions, anomaly detection methods have been inspected
to provide additional protection \cite{CBK09,GT09}. The main idea is to
detect the presence of an adversary inside an organization's network, based
on the adversary's actions when it moves from one spot to another, or tries
to access sensitive data (honeypots). Often, the detection rests on log file
analysis, with data collected from all over the network and applications
therein. Designated logging engines (e.g., syslog
(http://tools.ietf.org/html/rfc5424) or logging management solutions (e.g.,
Graylog (http://graylog2.org/) are usually in charge here. Nevertheless, the
detection of exceptional events in these log files alone is insufficient,
since anomalies are often exposed not before events are correlated with each
other \cite{KTK02,HS09}. Since today's systems are heavily connected and
interchange a large amount of data on a regular basis, the size of the
logging information increases drastically, making an evaluation quite
difficult.

An example for a tool realizing this approach is AECID (Automatic Event
Correlation for Incident Detection) \cite{SF13,SFF14}. AECID enforces
white-lists and monitors system events, their occurrences as well as the
interdependencies between different systems. In the course of this, the
system is able to get an overview on the ``normal'' behavior of the
infrastructure. If some systems start to act differently from this normal
behavior, an attack is suspected and an alert is raised.

Whereas AECID (and similar tools) are detective measures (as they trigger
alerts based on specific events that have happened already), our approach in
the following is \emph{preventive} in the sense of estimating and minimizing
the risk of a successful \ac{APT} from the beginning (cf.
Section \ref{sec:our-contribution}). Game theory is here applied to optimize the
defense against a stealthy invader, who attempts to sneak into the system on
a set of known paths, while the defender does its best to guard all these
ways simultaneously. This is the abstract version of the situation that is
normally summarized under the term \ac{APT}.

Game theory appears as a natural tool to analyze conflicts of interest, such
as obviously arise between the defender and the attacker mounting an
\ac{APT}. Powerful techniques to defend against stealthy takeover have been
defined (partially originating from \cite{Dijk2013,Zhang2015} but also based
on a variety of precursor and independent approaches, such as collected in
\cite{Alpcan2010}), but a method that fits into established risk management
processes and can be instantiated with vague, fuzzy and qualitative risk
assessments (such as uttered by domain experts) is demanding yet missing.
Particularly intricate are matters of social risk response, say, if an
enterprise seeks to minimize losses of reputation besides direct costs;
assessing the public community's response to certain actions being taken is a
vague and difficult issue, to which sophisticated game theoretic
\cite{Xia&Meng&Wang2015,Meng&Xia2015,Chen&Wang2015,Chen&Wang2016} and agent
based models \cite{Busby&Onggo&Liu2016} can be applied for an analysis and
risk quantification. A recognized feature of any game-theoretic treatment of
\ac{APT} and in general every cyber-security scenario is the lack and
asymmetry of information (say, the absence of knowledge about the attacker's
strategy spaces or payoffs, cf. \cite{Nguyen2009,Pavlovic2011}, while the
attacker may have full information about the target system). This asymmetry
is even stronger than what can be captured by many game-theoretic models,
since organizational constraints may enforce the defender to act only at
certain points in time, while the attacker is free to become active at any
time. That is, the game is \emph{discrete time} for one player, but
\emph{continuous time} for the other player -- a setting that is hardly
considered in game-theoretic literature related to security, and as such a
central novelty in this work.

As we will show later (cf. Section \ref{sec:continuous-time-actions} and
Lemma \ref{lem:zero-sum-optimality}), matrix games are nonetheless a proper
model to account for what the \emph{defender} can do against an \ac{APT}, if
we confine ourselves with the goal of playing the game to the best of our own
protection and allow the outcomes to be random and unpredictable. Under this
relaxation over the conventional game theoretic modelling, we can account for
the outcome to be dependent on an action that is taken at different points in
time, and especially also for actions that were interrupted before they could
carry to completion. This addresses the issue identified by
\cite{Hamilton2002}, who pointed out that moves may take a variable amount of
time rather than being instantaneous (and thus atomic).

Ultimately, a significant obstacle for practitioners in the application of
any game theoretic model is the lack of understanding of the ingredients to
the game. That is, no matter how sophisticated the model may be, it
nevertheless needs to be instantiated with whatever data is available. In
many cases, this data is either qualitative (fuzzy) expert knowledge
(formulated in some taxonomy, e.g., \cite{Innerhofer-Oberperfler2009}) or
obtained from simulation (see \cite{Wellman2014} for one example). Either may
not be suitable to instantiate the proper \ac{APT} model, even though the
\ac{APT}-game model would be quite sophisticated and powerful (such as
\cite{Zhu2013}) in its capabilities for risk mitigation. In any case, this
takes us to \emph{empirical game-theoretic models}; a category into which
this work falls.

\subsection{Our Contribution} \label{sec:our-contribution}
We present a novel form of capturing payoff uncertainty in game theoretic
models. We deviate from standard games in the conceptual way of measuring the
outcome of a gameplay not in crisp terms, but by an entire probability
distribution object. That is, we play game theory on the abstract space of
distributions for the following reasons:
\begin{enumerate}
  \item The specification of losses and payoffs in a game is often
      difficult: how would we accurately quantify the results of a defense
      in light of an attack? Do we count the number of infected machines
      (such as done in \cite{Dijk2013})? Shall we work with monetary loss
      (causing difficulties in how to ``price-tag'' loss of reputation or
      consumer's trust)? Conversely, can we play games over a categorical
      scale of payoffs,
  such as risk is being quantified in many standards like ISO 31000 \cite{ISO31000} or
  similar?

  We can elegantly avoid any such issue by letting the game being defined
  by any outcome that can be ordered (as in conventional game theory), but
  in addition, allowing an action to have many different random outcomes
  (this is usually not possible in standard games). In doing so, we gain a
  considerable flexibility and degree of freedom to tackle a variety of
  issues, which we will discuss later on.
  \item There is a strong asymmetry in the player's information in many
      senses: first, the game structure itself is not common knowledge,
      since the defender knows only little about the opponent, while the
      opponent knows very much about the defender's infrastructure (as lies
      in the nature of \acp{APT}, since these typically include an a-priori
      phase of investigation and espionage). Second, the game play is
      different for both players, since moves are not mutually observable,
      nor must happen instantaneously or even at the same times.

      Again, this can be captured by letting the effects of action be
  nondeterministic and random even if both, the attacker's and defender's
  action were both known.
  \item Any game-theoretic model for security may itself be only part of an
      outer risk management process, and as such must be ``compatible''
      with the surrounding workflows, which cover \ac{APT} mitigation among
      other aspects. That is, the game theoretic model's input and output
      must be useful with what the risk management process can deliver and
      requires. Our \ac{APT}-games will be designed to fulfil this need.

  \item Conventional stochastic models like Bayesian games indeed also
      capture uncertainty, but do so by letting the modeler describe a
      variety of different possible game structures, among which nature
      chooses at random in the actual gameplay. While different such
      structures can embody different outcomes, and the likelihood for
      these can be specified as a distribution (similar to what we do),
      each of these possible game structures must be specified in the
      classical way, thus effectively ``multiplying'' the problems of
      practitioners (if one game is difficult to specify, the specification
      of several ones does not appear to ease matters). Our approach avoids
      these issues by working with empirical data directly, and keeping the
      game models simple at the same time.
\end{enumerate}

In light of the last point in particular, we will restrict our attention in
the following to the problem of how to define games over qualitatively
assessed outcomes that may be random. That is, the central question that this
work discusses is essentially a form of reasoning under uncertainty:

\begin{quote}
\textbf{Given some possibilities to act, what would be the best choice if
the consequences of an action are intrinsically random?}
\end{quote}

We will show how to answer this question if the randomness can be modeled in
the most general form by specifying probability distributions for the
outcome. However, unlike normal optimization that maximizes some numeric
quantity derived from the distribution of a random variable $X$, our games
will optimize the \emph{shape} of $X$'s distribution itself.

The presentation will heavily use examples for illustration, yet the concepts
themselves will be described and also defined in a general form. To get
started, consider the following example of decision making under the setting
that we consider. Example \ref{exa:ipr-responsibilities} is about the
protection of \ac{IPR}, whose theft can be a reason to mount an \ac{APT}.

\begin{exa}[Assigning \ac{IPR} Responsibilities]\label{exa:ipr-responsibilities}
Assume that an enterprise runs a project and is worried about protection of
\ac{IPR}. To mitigate this issue, one or more persons shall be put in charge
of \ac{IPR} protection. For this, say, three options are available:
\begin{enumerate}
  \item \emph{Assign \ac{IPR} responsibility to one person}: This will
      increase the workload of the employee, and must be made w.r.t.
      available resources and skills. Neither is precisely quantifiable nor
      may be sufficient at all times. Thus, even assuming a strong
      commitment of the person to its role, some residual risk of damage
      occurring remains (human error of subordinates cannot be ultimately
      ruled out despite any strong supervision).
  \item \emph{Assign \ac{IPR} responsibility to a team of two or three
      persons}: Resources and skills may be much richer in this setting,
      but there is a danger of mutual reliance on one another, such that in
      the worst case, no-one really does the job (as a result of social
      coordination failure). Chances for this worst case to occur may be
      even higher than for option 1.
  \item \emph{Do security/\ac{IPR} training sessions}: here, we would
      completely rely the joint behavior of the employee's and their
      commitment to the confidentiality of project content and adherence to
      the training's messages. Nevertheless, chances for \ac{IPR} loss due
      to human error (e.g., an unencrypted email leaking confidential
      information, or similar) may be lowered only temporarily, so that the
      training would have to be repeated from time to time.
\end{enumerate}
\end{exa}

The optimal choice in Example \ref{exa:ipr-responsibilities} is not obvious,
since consequences are not all entirely guaranteed for always foreseeable.
Intuitively, we would go with the setting under which loss of intellectual
property is least likely. Later, in Section \ref{sec:uncertainty-modeling},
we will construct an ordering relation $\preceq$ that does exactly this if
the loss distributions are defined on a scale of damage whose maximum is the
loss of intellectual property (e.g., quantified by the business value
attached to it). So, if we are somehow able to model possible random outcomes
in each of the three scenarios, we can (algorithmically) compute the ``best''
(i.e., $\preceq$-minimal) loss distribution to be the best choice among the
three above. This is a matter of loss distribution specification, which we
will discuss in Section \ref{sec:loss-specifications}.

Finally, we remark that all of the theory sketched here has been implemented
in \texttt{R} (including full fledged support of multi-criteria game theory
based on distributions as discussed in section \ref{sec:tradeoffs}), to
validate the method and to compute the results for the examples (such as in
Section \ref{sec:concrete-example}) shown here.

\subsection{Organization of the paper}
Section \ref{sec:preliminaries} briefly introduces the tools and concepts
that our models are built upon. We will strongly rely on \ac{TVA} (see
Section \ref{sec:tva}) and human expertise in our model building, which we
believe to be a viable approximation of how security risk management works in
practice. Section \ref{sec:running-example} presents an example, which we
will carry through the article to illustrate the concepts and approach as a
whole. Section \ref{sec:uncertainty-modeling} introduces the theory of how
decisions can be made if their outcome is rated through an entire probability
distribution object (rather than a number), and Section
\ref{sec:practical-decision-making} takes this basis to define games,
equilibria and to highlight similarities but also an important qualitative
difference between the so-generalized games and their classical counterparts (among others, \cite{burgin_remarks_2021} 
showed that fictitious play can converge to a point that is not necessarily an equilibrium, which led to the introduction of a lexicographic Nash equilibrium as a new concept in \cite{rass_game_2021}; we will go into more details later in Section \ref{sec:practical-decision-making}).
Sections \ref{sec:apt-games} and \ref{sec:practicalities} apply the framework
to \acp{APT} by picking up the example from Section
\ref{sec:running-example}, and give algorithmic details on how to practically
work out results (the aforementioned differences between our and classical
games call for various mathematical tricks here). Section
\ref{sec:generalizations} briefly discusses generalizations towards
multi-criteria decision making. Section \ref{sec:concrete-example} finishes
the example from Section \ref{sec:running-example} by presenting results and
security protection advices obtained from our game-theoretic \ac{APT}
mitigation game. Section \ref{sec:discussion} presents a critical discussion
in terms of answering direct questions that were collected from practical
experience with the proposed method in a research project (see the
acknowledgment section at the end of the paper). Conclusions are drawn in
Section \ref{sec:conclusions}.


\section{Preliminaries and Notation}\label{sec:preliminaries}
Vectors and matrices will be denoted as bold-face letters in lower case for
vectors and upper case for matrices. $n\times m$-Matrices over a set $M$ are
denoted as $\vec A\in M^{n\times m}$, and the symbol
$(a_n)_{n\in\N}=(a_n)_{n=1}^\infty$ is a shorthand for sequences. We use
upper case normal font letters like $X$ to denote \acp{RV}, and write $X\sim
F$ to express that the \ac{RV} $X$ has the distribution function $F$. The
respective density belonging to $F$ is the corresponding lower case letter
$f$, and where necessary, we add the subscript $F_X$ or $f_X$ to indicate the
related \ac{RV} for the density or distribution. For a given (finite) set
$M$, we let $\S(M)$ be the set of all discrete probability distributions (the
simplex) over $M$. Likewise, all families of sets, \acp{RV} or distribution
functions are denoted in calligraphic letters (such as $\UF, \S$ or $\F$).
Estimates of a value are indicated by a hat, such as $\hat F, \hat f$, to
mean empirical distributions (normalized histograms). Approximations of an
object $x$ or $F$ (scalar, distribution, etc.) are marked by a tilde, e.g.,
$\tilde x, \tilde F$.

\subsection{Topological Vulnerability Analysis}\label{sec:tva}
Topological vulnerability analysis \cite{Jajodia2005} is the systematic
identification of attacks to a system, based on the system's structure and
especially its network topology. The process usually consists of creating a
complete picture of the infrastructure augmented by all available details
about the components. Modeling the system's topology as an (undirected) graph
$G(V,E)$ with a designated target node $v_0\in V$, we can use standard path
searching algorithms to identify paths from the exterior of $G$ towards the
target node $v_0$. Whatever structure is dug up by the TVA, an immediate
question concerns the applicability of known attack patterns to the
infrastructure model (similar to virus patterns being looked up in software).
Graph matching techniques (see
\cite{EmmertStreib&Dehmer&Shi2016,Dehmer&EmmertStreib&Shi2014} for example)
appear as an interesting tool to apply here. The known (or suspected)
vulnerabilities/exploits related to the nodes in $V$ then determines which
paths are theoretically open to $v_0$ to successfully attack the system.
These \emph{attack paths} are thus sequences of vulnerabilities (augmented
with the respective preconditions to exploit a vulnerability), and are the
main output of a \ac{TVA}. An \ac{APT} can then (in a slightly simplified
perspective) be viewed as the entirety of attack paths, and is physically
mounted by sequentially working along a chosen attack path in a way that
avoids detection of the attack at all stages (stealthy). Particular practical
risk arises from exploits of not yet known vulnerabilities, which are
commonly called \emph{zero-day exploits}. Uncertainty about these partially
roots in the complexity of the network, so that graph entropy measures (see,
e.g., \cite{Cao&Dehmer&Shi2014}) may be considered as a measure to help
quantifying the chances of attacks coming over paths that were missed during
the analysis. The practical handling of this residual risk is often a matter
of using domain knowledge, collecting expert opinions, experience and
information mining, combined with suitable mathematical models (e.g.,
\cite{Moore2010,Wang2014}). Our model immanently includes a zero-day
vulnerability measure, as will be discussed in Section \ref{sec:zero-day}.

\subsection{Attack Graphs and Attack Trees}
The entirety of ways into a system, including intersections and alternative
routes on attack paths makes up the \emph{attack graph} \cite{Noel2010}. It
is essentially a representation based on the system topology $G$, in which
outgoing links of a node $v$ are retained only if some exploit on $v$ enables
reaching $v$'s neighbor (see Figs \ref{fig:apt-infrastructure} and
\ref{fig:apt-example-tva} for an example). In terms of representation, an
attack graph is to be distinguished from an attack tree, which is usually an
AND/OR tree representing the possible exploit chains in a different way.
Regardless of which is available, the main object of interest for our
purposes is the set of attack paths, which directly corresponds to the action
set of player 2 in our \ac{APT}-games.

\subsection{Extensive Form Games}
Towards a game theoretic model of \ac{APT}, we will use \acp{EFG}. While a
full fledged formal definition of \ac{EFG} is lengthy and complex, it will
suffice to give a description of it to highlight the similarities to
\acp{APT}. Formally, \ac{EFG} are described by a tree $T(V_T,E_T)$, with a
designated root node that represents the starting stage in the game.
Consequently, $T$ has edges directed outwards from the root. For an \ac{APT},
the root corresponds to the (hypothetical) point representing the exterior of
the network graph $G$. The \ac{EFG} is played between a set of (for our
purposes two) players, including a hypothetical player ``chance'' that
represents random moves in the game. Each node $v\in V_T$ in the game tree
$T(V_T,E_T)$ carries an information on which player is currently at move
(including the ``chance'' player). Furthermore, moves that are
indistinguishable by other players are collected in a player's
\emph{information set}. This, from the opponent's perspective, represents the
uncertainty about what a player has currently done in the game. In an
\ac{APT} model, the information set would correspond to possible locations
where the attacker could currently be (again, recalling that an \ac{APT} is
stealthy). The \ac{EFG} description is completed by assigning a vector of
outcomes to the leaf nodes in the game tree. Normally, the outcomes are real
values and specified for all players. Viewing an \ac{APT} as an \ac{EFG}, we
would thus require to specify our own damage when the \ac{APT} has been
carried to the end (i.e., the target node $v_0$ has been reached), but also
the payoff to the adversary would needed to be known. The latter is a
practical issue, since the uncertainty in the game is not only due to the
attacker's moves themselves, but also caused by external influences outside
any of the player's influences. Shifting all this uncertainty induced
exteriorly to the chance-node in the \ac{EFG} description appears infeasible,
since much of it may depend on the particular action and current (even past)
moves of both players (defender being player 1 and the attacker being player
2 in an \ac{APT} game). However, given that the two players have different
information on the current stage of the game play (only the attacker knows
its precise position, the defender knows nothing; not even the presence of
the attacker is assured), defining the information sets appears hardly
doable.

Since the concept of \ac{EFG} as being games with imperfect information
essentially rests on information sets, any best behavior in the game play
will inevitably rest on hypotheses on the player's moves. For an \ac{EFG}, we
would describe these hypotheses as probability distributions on the
information sets. In lack of these, we can only model the outcome based on
the known defender's actions and assuming the attacker to be possibly
everywhere in the system. Practically, these hypotheses will rarely be
available as hard figures and mostly come in qualitative terms like ``low'',
``medium'' or ``high'' risk. Classical game theory is not naturally designed
to work in such fuzzy terms.

Finally, an assumption of frequent criticism concerns the game model to be
\emph{common knowledge} to both players. This is certainly questionable in
\ac{APT} scenarios. In fact, a defender may in practice only have limited and
widely uncertain information about the attacker's incentives, current moves,
current location or even its presence in the game. Thus, the information set
would in the worst case cover the entirety of the game graph, and neither is
the payoff to the attacker precisely quantifiable in most cases.

To avoid all these issues, we propose to replace the attacker's payoffs by
our own losses (in an implicit assumption of a zero-sum competition), in
which an equilibrium behavior is a provable bound (see Lemma
\ref{lem:zero-sum-optimality}) to the payoff for the player having modeled
the game . Second, we avoid difficulties of uncertain payoffs by defining the
game play itself as a one-shot event, in which both players choose their
strategies for the round of the game and the payoff is determined by that
choice (thus, shifting all matters of uncertainty about where a player is in
the game entirely to the payoffs). To this end, we will model an \ac{APT} as
a game with complete information but uncertain payoffs. In fact, the payoffs
will be entire probability distributions rather than numbers.

\section{A Running Example}\label{sec:running-example}

\begin{figure}[t!]
\includegraphics[scale=0.7]{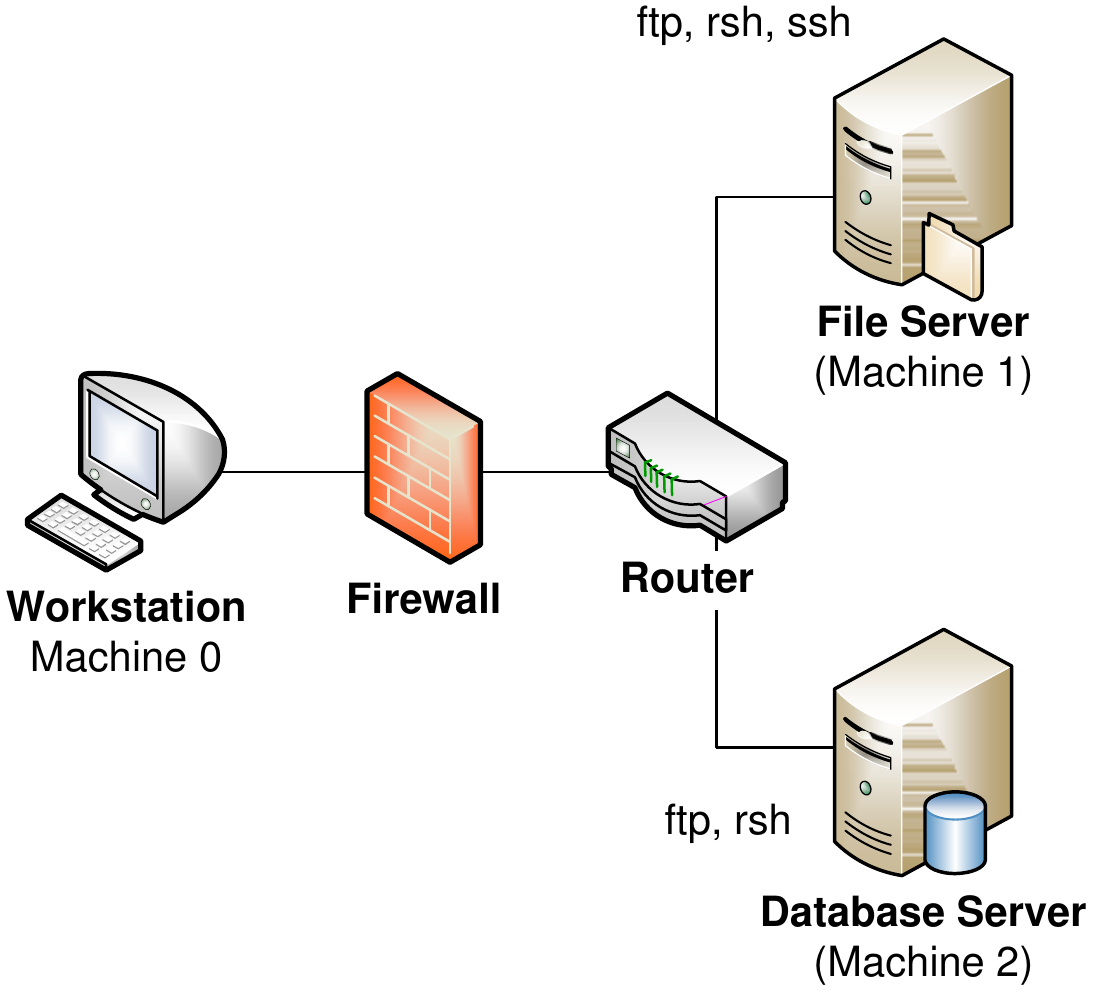}
\caption{\textbf{Infrastructure from \cite{Singhal2011} to illustrate game-theoretic \ac{APT} modeling}}\label{fig:apt-infrastructure}
\end{figure}

\begin{figure}[t!]
\includegraphics[scale=0.6]{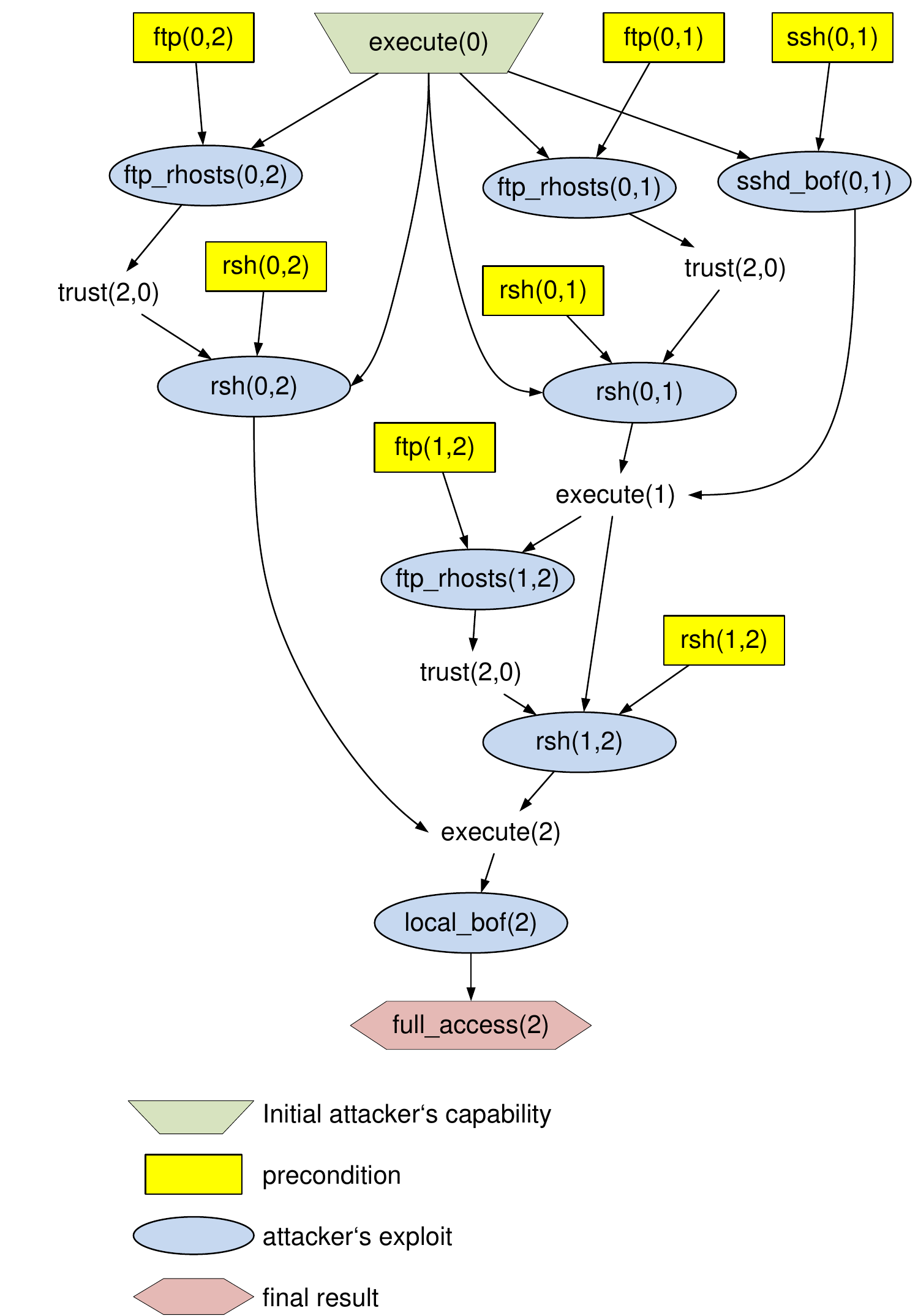}
\caption{\textbf{Example Attack Graph \cite{Singhal2011}}}\label{fig:apt-example-tva}
\end{figure}

Throughout this work, we will illustrate the steps and concepts using a
running example borrowed from \cite{Singhal2011}. This reference describes a
simple version of \ac{TVA} (see, e.g., \cite{Jajodia2005}) and attack graph
modeling, based on a small infrastructure that is shown in Fig
\ref{fig:apt-infrastructure}: The system consists of three machines (numbered
as 0, 1 and 2), with several services being open on each node (such as
\ac{FTP}, \ac{RSH} and \ac{SSH}). The adversary attempts to gain access to
machine 2, hereafter denoted as (the predicate) \texttt{full\_access(2)}.
Towards its goal, the attacker may run different exploits from various points
in the network, such as:
\begin{itemize}
  \item \ac{FTP}- or \ac{RSH}-connections from a node \texttt{x} to a
      remote host \texttt{y}, hereafter denoted as
      \texttt{ftp\_rhosts(x,y)}, and \texttt{rsh(x,y)}, respectively.
  \item a secure shell buffer overflow at node \texttt{y}, remotely
      initiated from node \texttt{x}, hereafter denoted as
      \texttt{sshd\_bof(x,y)}.
  \item local buffer overflows in node \texttt{x}, hereafter denoted as
      \texttt{local\_bof(x)}.
\end{itemize}
The actual \ac{APT} is the attempt to use these exploits (and combinations
thereof) in a stealthy fashion to penetrate the entire system towards
establishing full access to the target machine 2. Naturally, exploits of any
kind are subject to preconditions holding on the machine from which the
exploit is initiated. We denote such a precondition on machine \texttt{x} to
target a machine \texttt{y} in predicate notation as \texttt{ftp(x,y)},
\texttt{rsh(x,y)} and \texttt{ssh(x,y)}, w.r.t. the protocol being used.
Depending on which services are enabled and responsive on each machine, a
\ac{TVA} can then be used to compile an attack graph (see
\ref{fig:apt-example-tva}), which roots at the initial condition of having
execution privileges on machine 0, denoted as \texttt{execute(0)}, from which
attacks can be mounted under the relevant preconditions. A particular
\ac{APT} scenario can be viewed as a path in the graph that starts from the
root (\texttt{execute(0)}) via trust relations established between connected
machines \texttt{x} and \texttt{y} (denoted as \texttt{trust(x,y)}), until
the goal (\texttt{full\_access(2)}).

Running the plain task of computing and enumerating all paths in the attack
graph from \texttt{execute(0)} to \texttt{full\_access(2)} digs up 8 attack
vectors in our example. Each of these corresponds to one particular \ac{APT}
scenario, and the entirety of which makes up the adversary's \emph{action
set}, denoted as $AS_2$ (the subscript is used for consistency with the
subsequent game theoretic model, in which the attacker is player 2. The
defender will be player 1, respectively).

The next step in the risk mitigation process is the derivation of
countermeasures from the identified attacks, such as, for example, the
deactivation of services (to violate the necessary preconditions), or a
patching strategy (to remove buffer overflow vulnerabilities), to name only
two possibilities. Alas, none of these precautions is guaranteed to be
feasible or even to work, as for instance:
\begin{itemize}
  \item services may be vital to the system, say, deactivating an FTP
      connection may render the service offered by machine 1 useless.
  \item patches may work against a known buffer overflow, but an unknown
      number of similar exploits may nonetheless remain (thus enabling zero
      day attacks upon vulnerabilities found and offered for sale on the
      black market).
\end{itemize}
On the positive side, even unknown malware may be classified as such based on
heuristics, experience or innovative antivirus technologies, all of which
adds to the chances for the identified mitigation strategies to succeed. The
practical issue here is, however, to deal with the residual risks and the
inevitable uncertainty in the effectiveness of a protection. Ways to capture
and handle these issues are theoretically described in Section
\ref{sec:uncertainty-modeling} and applied to this example in Section
\ref{sec:apt-games}.

For the time being, let us assume that a (non-exhaustive) selection of
countermeasures has been identified and listed in Table \ref{tbl:as1}. We
call this list the \emph{defender's action set}, denoted as $AS_1$ to
indicate the defender as being player 1 in the subsequent \ac{APT} game
(Section \ref{sec:apt-games}). We leave this set incomplete here for the only
sake of simplicity (in reality, the analysis would dig up a much richer set
of countermeasures, such as can be based on the security controls catalog of
relevant norms as ISO 27001 \cite{ISO27001} or related).

\begin{table}[h!]
  \caption{\textbf{Security controls (selection)}}\label{tbl:as1}
  \begin{tabular}{|p{0.21\textwidth}|p{0.7\textwidth}|}
    \hline
\textbf{Countermeasure} & \textbf{Comment}\tabularnewline
\hline
deactivation of services (FTP, RSH, SSH) & these may not be permanently disabled, but could be temporar- ily
turned off or be requested on demand (provided that either is feasible
in the organizational structure and its workflows)\tabularnewline
\hline
software patches & this may catch known vulnerabilities (but not necessarily all of them),
but can be done only if a patch is currently available\tabularnewline
\hline
reinstalling entire machines & this wipes out unknown malware but comes at the cost of a temporary
outage of a machine (thus, causing potential trouble with the overall
system services)\tabularnewline
\hline
organizational precautions & for example, repeated security trainings for the employees. These
may also have only a temporary effect, since the security awareness
is raised during the training, but the effect decays over time, which
makes a repetition of the training necessary to have a permanent effect.\tabularnewline
\hline
  \end{tabular}
\end{table}

In general, the effect of an action, precaution, countermeasure, etc. is in
most cases not deterministic and influenced by external factors beyond the
defender's influence and not even fully determined by the attacker's actions.
Furthermore, actions on both sides are usually not for free, and costs/losses
on the defender's side are induced by system outages (say, during a
reinstall), staff unavailability (say, when people are in a training that itself
may be costly), etc. Some of these costs may be precisely calculated, but
others (say, if the system is offline during a reinstall) may depend on the
current workload and thus be difficult to quantify.

Therefore, a \emph{qualitative} risk assessment is often the only practical
option (and an explicit recommendation by various standards such as ISO 31000
and by the German Federal Office of Information Security (BSI)). For a game
theoretic analysis, however, this is inconvenient as it may result in a quite
vague assessment of a countermeasure that may look like shown in Table
\ref{tbl:security-precaution}.

\begin{table}[h!]
  \caption{\textbf{Example assessment of a security precaution}}\label{tbl:security-precaution}
  \begin{tabular}{|l|l|}
    \hline
    \multicolumn{2}{|l|}{\textbf{Countermeasure}: \emph{patching}}\\
    \textbf{Aspect} & \textbf{Expert's assessment} \\\hline
    applicability & not always available\\
    effectiveness & low or high (depending on the exploit) \\
    cost & low to medium (e.g., if the system needs to be rebooted) \\
    \hline
  \end{tabular}
\end{table}

Similar assessments can be made for other protective measures as well, with
quantitative figures occasionally being available (such as the costs for a
security training, or the cost to install a new firewall or intrusion
detection system). However, ambiguous and even inconsistent opinions may be
obtained on the effectiveness and applicability of a certain action. Even if
only one expert does the assessment in categorical terms as shown in Table
\ref{tbl:security-precaution}, uncertainty may at least arise from none of
the offered categories being appropriate for the real setting. That is, with
``medium effectiveness'' being a vaguely understood term in that context, an
expert may utter a range of possibilities rather than confining her/himself
to a specific statement. The example in Table \ref{tbl:security-precaution}
illustrates this by saying that the effectiveness of a patch can be either
high (if the patch closes precisely the buffer overflow that was intended by
the adversary), or even low, if the exploit has already been used to install
a backdoor, so that the buffer overflow -- even if it gets fixed -- is no
longer needed for the \ac{APT} to continue. What is even worse, both
assessments are at opposite ends of the scale (low/high), and can both be
justified, thus telling hardly anything informative in this case.

It is this point, where further opinions should be sought, which naturally
will create a number of different assessments, some of which may be even
mutually inconsistent (see Fig \ref{fig:expert-ratings} for an illustration
of how different opinions may accumulate at different points on the risk
scale).

All this hinders the application of conventional decision or game theory,
since in either approach (game or decision theoretic), we require a
reasonably measurable effect for an action, and also a way to uniquely rank
(order) different effects when a ``best'' action is sought.

\section{Modeling Uncertainty for Decision-Support}\label{sec:uncertainty-modeling}
If the outcomes of an action are uncertain, even random, then the most
powerful model to express these would be to:
\begin{itemize}
  \item collect as much data, expert opinions, etc. as is available,
  \item and compile a probability distribution from the available data, to
      capture the uncertainty in the assessments. Though this preserves all
      available information in the distribution object, the issue of
      \emph{working} with it is more involved and a central technical
      contribution in this work.
\end{itemize}
In the best case, the assessments turn out to be quite consistent, with only
a few outliers, in which case we may be able to define a reasonable
representative (say, the average assessment; see Fig
\ref{fig:expert-ratings}a). In other cases, however, the distribution may be
multimodal, with each peak corresponding to different answers that may all
have their own justification as being plausible (Fig
\ref{fig:expert-ratings}b shows an example of many experts agreeing on either
low or high effectiveness of the patching strategy).


\begin{figure}
  \centering
  \includegraphics[scale=1]{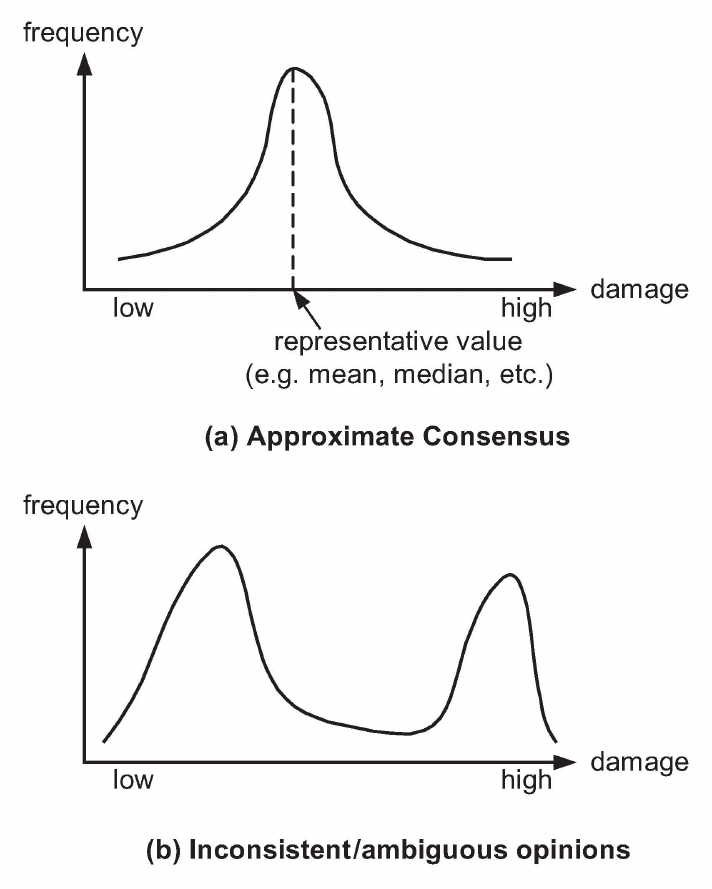}
  \caption{\textbf{Agreeing vs. disagreeing expert ratings}}\label{fig:expert-ratings}
\end{figure}

Finding a best action is typically done by assigning a utility value $u:
AS_1\times AS_2\to\R$ (see Section 2.2 in \cite{Robert2001}) to the actions
to choose from $AS_1, AS_2$, and looking for the maximal such utility for
both sides (defender and attacker). In quantitative risk management, a
popular choice for this utility value is the expected damage, computed as
\begin{equation}\label{eqn:risk}
    \text{risk} = \text{damage} \times \text{likelihood},
\end{equation}
which enjoys wide use throughout the literature (e.g., the ISO 31000
\cite{ISO31000} or ISO 27000 \cite{ISO27000} family of standards). This
convention is easily recognized as being the first moment of some (usually
not explicitly modeled) payoff distribution, and as such, is not satisfying
in practice, as the mean tells us nothing about possible variations about it
(Fig \ref{fig:different-preferences} illustrates the issue graphically). So,
the variance would be the next natural value to ask for in addition to
\eqref{eqn:risk}. Continuing this approach, we can describe a distribution
more and more accurately by using more and more moments, and indeed, the
mapping
\[
    \phi: F\mapsto \ell = (E(L^n))_{n\in\N}\in\R^{\infty}
\]
provides a bijective link between a distribution function $F$ and a
representative infinite sequence of real numbers, provided that all moments
exist. In a simple use of this representation, we could just
lexicographically compare the sequences, letting the first judgement be based
on the mean, and in case of equality, compare the variances, and so on. Such
an ordering, however, appears undesirable in light of easy to construct
examples that yield quite implausible preferences. Fig
\ref{fig:different-preferences} shows an example.

\begin{figure}
  \includegraphics[scale=1]{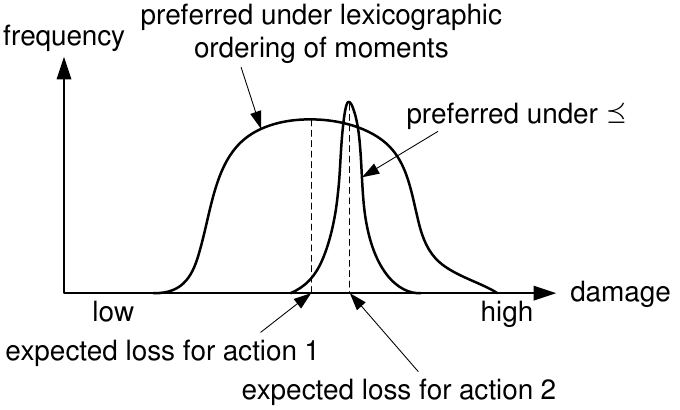}
  \caption{\textbf{Comparing Different Preference Rules}}\label{fig:different-preferences}
\end{figure}

An approach that preserves all information is treating the moment sequence as
a hyperreal number $\phi(F)\in\HR$, so that we get a ``natural ordering'' on
the distributions as it exists in the hyperreal space $(\HR,\leq)$; see
\cite{Robinson1966} for a full detailed treatment, which we leave here as
being out of the scope of this work.

Nevertheless, it is important to recognize that the trick of embedding a
distribution in the ordered field $(\HR,\leq)$ of hyperreals equips us with a
\emph{full-fledged arithmetic} applicable to random payoff distributions, as
well as a \emph{stochastic ordering}, so that ``optimality'' of decisions can
be defined soundly (later done in Definition \ref{def:preference}). This
implies that many well known and useful results from game and decision theory
remain applicable in our setting (almost) as they are. Essentially, this
saves us the labour of re-establishing a lot of theory, as would be necessary
if another stochastic order (such as one in \cite{Shaked2006}) would be used.

Definition \ref{def:loss} will suitably restrict the class of loss
distributions to ensure that all moments exist. Before that, however, let us
briefly recap where the loss distributions will come from:

Given the attack graph that describes all \ac{APT} scenarios and treating it
as an \ac{EFG} game description, we apply the same conversion of an \ac{EFG}
into the normal form of the game, which is a matrix. Let $n=\abs{AS_1},m =
\abs{AS_2}$ be the number of threat mitigation strategies and possible
exploits, which correspond to the action sets of both players (paths through
the infrastructure determined by the possible exploits; cf. Table
\ref{tbl:as2}). Whereas a classical game would be described as a real valued
payoff matrix $\vec A\in\R^{n\times m}$, the outcome in the \ac{APT} game is
not deterministic and as such will be described by a matrix of \acp{RV} $\vec
A=(L_{ij})_{i,j=1}^{n,m}$. Each variable $L_{ij}$ describes the random loss
(effect) of taking mitigation strategy $i$ relative to the unknown $j$-th
move of the adversary.

\begin{table}[h!]
\small
  \caption{\textbf{\ac{APT} scenarios (adversary's action set $AS_2$, based on Fig \ref{fig:apt-example-tva})}}\label{tbl:as2}
  \begin{tabularx}{\columnwidth}{|l|X|}
    \hline
    1 & \texttt{execute(0)} $\to$ \texttt{ftp\_rhosts(0,1)} $\to$ \texttt{rsh(0,1)} $\to$ \texttt{ftp\_rhosts(1,2)} $\to$ \texttt{rsh(1,2)} $\to$ \texttt{local\_bof(2)} $\to$ \texttt{full\_access(2)} \\\hline
    2 & \texttt{execute(0)} $\to$ \texttt{ftp\_rhosts(0,1)} $\to$ \texttt{rsh(0,1)} $\to$ \texttt{rsh(1,2)} $\to$ \texttt{local\_bof(2)} $\to$ \texttt{full\_access(2)}\\\hline
    3 & \texttt{execute(0)} $\to$ \texttt{ftp\_rhosts(0,2)} $\to$ \texttt{rsh(0,2)} $\to$ \texttt{local\_bof(2)} $\to$ \texttt{full\_access(2)} \\\hline
    4 & \texttt{execute(0)} $\to$ \texttt{rsh(0,1)} $\to$ \texttt{ftp\_rhosts(1,2)} $\to$ \texttt{rsh(1,2)} $\to$ \texttt{local\_bof(2)} $\to$ \texttt{full\_access(2)} \\\hline
    5 & \texttt{execute(0)} $\to$ \texttt{rsh(0,1)} $\to$ \texttt{rsh(1,2)} $\to$ \texttt{local\_bof(2)} $\to$ \texttt{full\_access(2)} \\\hline
    6 & \texttt{execute(0)} $\to$ \texttt{rsh(0,2)} $\to$ \texttt{local\_bof(2)} $\to$ \texttt{full\_access(2)} \\\hline
    7 & \texttt{execute(0)} $\to$ \texttt{sshd\_bof(0,1)} $\to$ \texttt{ftp\_rhosts(1,2)} $\to$ \texttt{rsh(1,2)} $\to$ \texttt{local\_bof(2)} $\to$ \texttt{full\_access(2)} \\\hline
    8 & \texttt{execute(0)} $\to$ \texttt{sshd\_bof(0,1)} $\to$ \texttt{rsh(1,2)} $\to$ \texttt{local\_bof(2)} $\to$ \texttt{full\_access(2)} \\
    \hline
  \end{tabularx}
\end{table}

\begin{defn}[(Random) Loss]\label{def:loss}
A real-valued \ac{RV} $L$ is called a \emph{(random) loss}, if the following
conditions are satisfied:
\begin{itemize}
  \item $L\geq 1$ (this can be assumed w.l.o.g.)
   \item The support of $L$, being $\supp(L):=\overline{\set{x\in\R:
       f_L(x)>0}}$, is bounded (where the bar means the topological
       closure).
   \item $L$ has a density $f_L$ w.r.t. either the counting- or the
       Lebesgue-measure. In the latter case, we assume the density $f_L$ to
       be continuous and piecewise polynomial over a finite partition of its support.
\end{itemize}
Define the \emph{set of loss distributions} $\F$ to contain all distribution
functions related to random losses.
\end{defn}
We stress that the conditions of Def. \ref{def:loss} only mildly constrain the set of choices, since (by Weierstra\ss' theorem), all smooth distributions allow a polynomial approximation in the desired sense, if they are compactly supported.

\subsection{Optimal Decisions if Consequences are
Uncertain}\label{sec:optimal-decisions}

Definition \ref{def:loss} assures that the density function $f_L$ of any
random loss $L$ admits moments $E(L^n)$ of all orders $n\in\N$, so that we
get a well-defined condition for an ordering based on moment sequences:

\begin{defn}[$\preceq$-Preference between Loss Distributions]\label{def:preference}
Let $L_1\sim F_1, L_2\sim F_2$ be losses
with $F_1, F_2\in\F$. We prefer $L_1$ over $L_2$, written as $L_1\preceq
L_2$, if there is an index $k_0\in\N$ so that $E(L_1^k)\leq E(L_2^k)$ for all
$k\geq k_0$. We synonymously write $F_1\preceq F_2$ whenever we explicitly
refer to distributions rather than \acp{RV}.
\end{defn}

A minor technical difficulty arises from the yet unsettled issue of whether
or not there are non-isomorphic instances of $(\HR,\leq)$, in which case we
could get ambiguities in the $\preceq$-ordering. The next result, however,
rules out this danger.

\begin{prop}[cf. \cite{Rass2015b,rass_game_2021}]
The set $(\F,\preceq)$ with $\F$ as in definition \ref{def:loss} and
$\preceq$ as in definition \ref{def:preference} is a totally ordered set,
where $F_1\preceq F_2$ implies $\phi(F_1)\leq\phi(F_2)$, with the embedding
$\phi:\F\to(\HR,\leq)$, and the $\preceq$-ordering on $\F$ is invariant
w.r.t. how $(\HR,\leq)$ is constructed.
\end{prop}

While the theoretical definition is easy, important practical questions about
this preference demand an answer, in particular:
\begin{enumerate}
  \item What is the practical meaning of the $\preceq$-ordering for risk
      management?
  \item If $\preceq$ is practically meaningful, how can we (efficiently)
      decide it?
\end{enumerate}

Let us postpone the answer to the first question until Section
\ref{sec:practical-meaning-of-preference}, and come to the algorithmic
matters of deciding $\preceq$ first. The answer to the second question will
then also deliver the answer to the first one.

\subsection{Practical Decision of $\preceq$-Preferences}\label{sec:practical-decision-of-preferences}
Let us first discuss the case where the loss distribution is continuous.
Common examples in risk management (cf. \cite{Embrechts2003}) are extreme
value distribution or stable distribution (with fat tails). Although such
distributions may not necessarily have a bounded support (thus not
corresponding to a random loss in the sense of definition \ref{def:loss}), we
can approximate the distributions by random losses via defining a \emph{risk
acceptance threshold} $1<a\in\R$, and truncating the distribution outside the
range $[1,a]$. The concrete value $a$ can be chosen upon a desired accuracy
$\eps>0$, for which we can choose $a$ large enough to have the residual
likelihood of damage $>a$ is smaller than $\eps$, or formally,
$\Pr(L>a)<\eps$ (we will come back to the choice of $a$ in section
\ref{sec:faqs}).

Practically, the risk acceptance threshold is the value above which risks are
simply ``being taken'' or are covered by proper insurance contracts. Thus,
specifying the value $a$ and truncating the loss distributions accordingly
makes distributions with fat and/or unbounded tails fit as approximate
versions into definition \ref{def:loss}.

If $L_1,L_2$ have the same compact support $[1,a]\subset\R$, and since the
respective density functions  $f_{L_1}, f_{L_2}$ are assumed continuous, both
admit limits $b_1=\lim_{x\to a}f_{L_1}(x)$ and $b_2=\lim_{x\to a}f_{L_2}(x)$.
For the moment, assume $b_1\neq b_2$, i.e., $f_{L_1}(a)\neq f_{L_2}(a)$ (the
case of equality is treated later). The continuity of both functions implies
that $f_{L_1}(x)\neq f_{L_2}(x)$ holds in an entire left neighborhood
$(a-\eps,a]$ of $a$ for some $\eps>0$. It is then a simple matter of calculus
to verify that (since both $L_1,L_2\geq 1$), the speed of divergence of the
respective moment sequences $(E(L_1^n))_{n\in\N}$ and $(E(L_2^n))_{n\in\N}$
is determined by which density function takes larger values in the region
$(a-\eps,a]$, recalling that both densities vanish at $x>a$. That is, we have
\begin{equation}\label{eqn:moment-divergence}
\lim_{n\to\infty}[E(L_1^n)-E(L_2^n)]\in\set{-\infty,+\infty},
\end{equation}
and the condition of Definition \ref{def:preference} is ultimately satisfied
(in either way).

\begin{lem}\label{lem:preference-decisions}
Let $L_1,L_2$ be two random loss variables with continuous distribution
functions $F_1,F_2\in\F$, and let $f_1, f_2$ denote the respective densities.
If both \acp{RV} are supported in an interval $[1,a]$ for $a\in\R$, and there
is some $\eps>0$ such that $f_1(x)<f_2(x)$ for all $x\in(a-\eps,a]$, then
$L_1\preceq L_2$.
\end{lem}

Lemma \ref{lem:preference-decisions} (see \cite{Rass2015b} for a proof)
offers an easy way to decide preferences based on the \ac{RV}'s density
functions only. The procedure is the following: Call $[1,a]$ the common
support of both loss variables $L_1,L_2$, and consider the density functions
$f_1,f_2$:
\begin{itemize}
  \item If $f_1(a)<f_2(a)$, then $L_1\preceq L_2$,
  \item Otherwise, if $f_1(a)>f_2(a)$, then $L_2\preceq L_1$.
\end{itemize}
Upon a tie, i.e., $f_1(a)=f_2(a)$, we need to either decrease $a$, truncate
the distributions properly, and repeat the analysis, or we may look at
derivatives at $a$ to tell us which density takes larger values locally near
$a$. The latter approach is further expanded in Section
\ref{sec:derivative-decisions}.

If the distribution is discrete, say, if the available data is not continuous
but qualitative (e.g., categorical), then things are even simpler: if
$L_1,L_2$ are both distributions over the same categories, then $L_1\preceq
L_2$, if $L_1$ puts less likelihood to categories of large damage than $L_2$
(see Fig \ref{fig:preference-illustration} for an example).

Formally, $\preceq$ thus boils down to a humble lexicographic ordering
whenever the losses have categorical distributions.
\begin{defn}[lexicographic ordering] For two vectors $\vec x =
(x_1,x_2,\ldots)$ and $\vec y=(y_1,y_2,\ldots)$ of not necessarily the same
length, we define $\vec x<_{lex} \vec y$ if and only if there is an index
$i_0$ so that $x_{i_0}<y_{i_0}$ and $x_i=y_i$ whenever $i<i_0$.
\end{defn}
For two categorical distributions given in matrix notation and letting the
support be given in descending order of risk levels $r_n>r_{n-1}>\ldots>r_1$,
we observe that $F_1\preceq F_2\iff (p_n,\ldots,p_1)=\vec p<_{lex} \vec
q=(q_n,\ldots,q_1)$, when the distributions are:
\begin{align*}
    F_1:& \left(
           \begin{array}{ccccc}
             p_n &    & \ldots &   & p_1 \\
             r_n & >  & \ldots & >  & r_1 \\
           \end{array}
         \right),\text{ and }
    F_2: \left(
           \begin{array}{ccccc}
             q_n &    & \ldots &   & q_1 \\
             r_n & >  & \ldots & >  & r_1 \\
           \end{array}
         \right)
\end{align*}
That is, the action with the higher likelihood of extreme damage is less
favorable, and upon a tie (equal chances of large damages), the likelihood
for the next smaller risk level tips the scale, etc. Rephrasing the classical saddle-point condition in terms of such a lexicographic order leads to a new concept that coincides with (standard) Nash equilibria only in the 1-dimensional case (as observed by \cite{burgin_remarks_2021}). In higher dimensions, corresponding to non-scalar losses, such as categorical loss distributions, an optimum w.r.t. the lexicographic order, does not necessarily also induce a saddle point in the sense of $\preceq$. We will hence use different names to distinguish classical from lexicographic equilibria, later in Section \ref{sec:practical-decision-making} and onwards.

\subsection{Practical Meaning of
$\preceq$-Preferences}\label{sec:practical-meaning-of-preference}
Summarizing the previous discussion in concise form directly takes us to the
practical meaning of $\preceq$-preferences:
\begin{quote}
\emph{We have $L_1\preceq L_2$, if large damages (near the maximum) are
more likely to occur under $L_2$ than under $L_1$. }\end{quote} This is
just an intuitive re-statement of Lemma \ref{lem:preference-decisions}.
However, and remarkably, the converse to it is also true, if the density is piecewise polynomial (see \cite{rass_game_2021,burgin_remarks_2021} as an extension to the original Thm. 2.14 in \cite{Rass2015b}):
\begin{thm}\label{thm:preference-meaning}
Let $L_1,L_2$ be two \acp{RV} with distribution functions $F_1,F_2\in\F$. If
$L_1\preceq L_2$, then a threshold $x_0$ exists such that
$\Pr(L_1>x)<\Pr(L_2>x)$ for every $x\geq x_0$.
\end{thm}

Restating this intuitively again, Theorem \ref{thm:preference-meaning}
tells that:
\begin{quote}
\emph{a $\preceq$-minimal decision among two choices with respective
consequences $L_1$ and $L_2$ minimizes the chances for large damages to
occur.}
\end{quote}

This is exactly what we are looking for: Risk management is in many cases
focusing on extreme events rather than small distortions (which the
system's ``natural'' resilience is expected to handle anyway), and the
focus of the $\preceq$-relation to prefer distributions with lighter tails
perfectly accounts for this. Having $\preceq$ as a total ordering with a
practical interpretation as being ``risk-averse'', this already addresses
the simple case of decision making among finitely many choices (as
discussed in the next Section \ref{sec:practical-decision-making}).

\section{Practical Decision-Making}\label{sec:practical-decision-making}
Remembering our example of \ac{APT} mitigation, suppose that as an initial
attempt, we would consider the installation of permanent security
precautions, such as (additional) firewalls, access controls, physical
protection, etc. Moreover, organizational changes such as were discussed in
example \ref{exa:ipr-responsibilities} may be under discussion. However, all
of these may have uncertain effectiveness, but the $\preceq$-relation now
helps out.

In general and abstractly, the decision problem and procedure is the
following:
\begin{itemize}
  \item A set of choices (e.g., security precautions)
      $d_1,d_2,\ldots,d_n\in AS_1$ is available (e.g., defense actions for
      \ac{APT} mitigation), each of which comes with a random
      consequence/effect captured by random losses $L_1,L_2,\ldots,L_n$.
  \item By looking for the $\preceq$-minimum among the distributions of
      $L_1,\ldots,L_n$, we can take an optimal decision under uncertainty.
\end{itemize}

An open issue so far is where to get the losses from, an issue that will be
revisited several times throughout this paper.

\subsection{Constructing Loss Distributions}\label{sec:loss-specifications}
The simplest approach to construct loss distributions that satisfy definition
\ref{def:loss} is to either:
\begin{itemize}
  \item collect as much data as is available, and compile an empirical
      distribution from it,
  \item or define the loss distribution directly based on expertise (say,
      if the action's incurred loss has a known distribution), if this is
      possible.
\end{itemize}
The latter case may occur seldom in practice, unless the particular threat
has been studied specifically (such as disaster management or value at risk
calls for extreme value distributions, etc.), and the ``adversary'' is nature
itself. Against a rational adversary such as business competitors, hackers,
etc., threat intelligence and expertise is the fundament upon which loss may
be measured. Often, this assessment is made in qualitative terms for several
(good) reasons, such as:
\begin{itemize}
  \item Human reasoning is neither numeric nor crisp, i.e., experts may
      find it simpler to give assessments like ``high risk'' instead of
      having to specify a hard figure.
  \item Numerical precision can create the illusion of accuracy where there
      is none. There are only few types of incidents on which reliable
      statistical data is available, and having huge amounts of data on
      \acp{APT} attacks on general cybersecurity incidents may be unrealistic (and
      also undesirable if the incidents concern oneself).
\end{itemize}

In practice, rating actions w.r.t. their outcomes is naturally a matter of
expert surveys, with answers possibly looking like shown in Table
\ref{tbl:security-precaution}. Collecting many such opinions and putting
them together in an empirical distribution $\hat L$ about a precaution's
performance may give distributions whose shape is unimodal (if a consensus
among opinions is found), or multimodal, if disagreeing opinions are
reported. Whatever happens or whether or not the outcome looks like
illustrated in Fig \ref{fig:expert-ratings}, the $\preceq$-preference
relation now allows for an elegant deal with this kind of uncertainty.

\begin{figure}
  \includegraphics[scale=0.6]{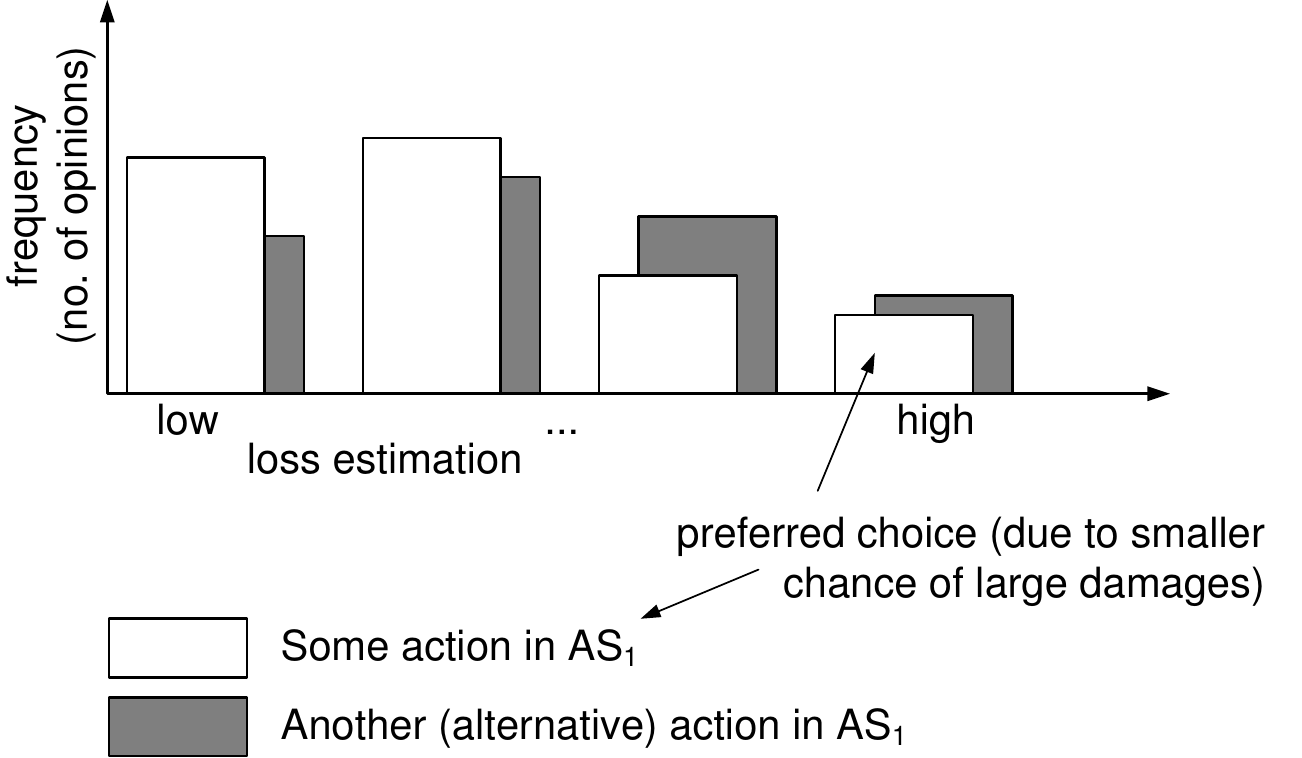}
  \caption{\textbf{Example of $\preceq$-choosing among two empirical distributions (inconsistent expert opinions)}}\label{fig:preference-illustration}
\end{figure}

\subsection{Games and Equilibria}
With the uncertain outcome in a scenario of defense $i$ vs. attack $j$ being
captured by a (perhaps empirical) probability distribution $L_{ij}$, and the
complete set of distributions being totally ordered w.r.t. $\preceq$, it is a
simple and straightforward manner to define matrix games and equilibria in
the well-known way, but will need to bear in mind that the resulting concepts will not exactly resemble (classical) equilibria in all senses, as we will explain later. For convenience of the reader, we give the necessary
concepts and definitions here from classical game theory.

Let $AS_1, AS_2$ be the action spaces for player 1 and 2, respectively,
with cardinalities $n$ and $m$. Let $\vec A=(L_{ij})_{i,j=1}^{n,m}$ be a
matrix of \acp{RV} that are all supported on the same compact set
$\Omega=[1,a]\subset\R$. Let $F_{ij}$ be the distribution function of the
random loss $L_{ij}$. In each round of the game, the random outcome $R$ is
conditional on the chosen actions of player 1 and player 2, and has the
distribution $R\sim L_{ij}$ if player 1 chooses action $i\in AS_1$ and
player 2 chooses action $j\in AS_2$.

We consider randomized choice rules $\vec p=(p_1,\ldots,p_n)\in \S(AS_1)$
and $\vec q=(q_1,\ldots,q_n)\in \S(AS_2)$, i.e., the vectors $\vec p, \vec
q$ describe the likelihoods of actions being taken by either player. In
that case, the random outcome has a distribution $R\sim F(\vec p,\vec q)$
computed from the law of total probability, which is
\begin{align}
     F(\vec p,\vec q)(r)=\Pr(R\leq r)&=\sum_{i,j}\Pr(R\leq r|i,j)\cdot \Pr(i,j)=\sum_{i,j}F_{ij}(r)\cdot p_i\cdot q_j,\label{eqn:utility-function}
\end{align}
assuming a stochastically independent choice of actions by both players.
This is the \emph{utility function} in case of random outcomes (note that
\eqref{eqn:utility-function} is exactly the same formula as is familiar
from matrix game theory). So, the actual gameplay is \emph{not} about
maximizing the average revenue (as usual in game theory), but towards
optimally ``shaping'' the outcome distribution $F(\vec p,\vec q)$ towards
$\preceq$-minimality. That is, in a zero-sum competition, player 1 and
player 2 seek to choose their actions in order to minimize/maximize the
likelihood of extreme events. Speaking differently again, player 1 attempts
to shift the mass allocated by the respective density $f(\vec p,\vec q)$
towards lowest damages, whereas player 2 tries his best to shape the
density $f$ towards putting more likelihood on larger damages. This is the
essential technical process of our game-theoretic \ac{APT} risk mitigation
strategies, whose optimality is that of a \emph{lexicographic Nash equilibrium} \cite{rass_game_2021}, which is, in the 1-dimensional case (only), the same as a standard Nash equilibrium (see \cite{burgin_remarks_2021} for a detailed example, and \cite{Fudenberg1991} for a formal
treatment of the classical case). In general, a lexicographic equilibrium respects goal priorities, which are here equal to the ordering on the loss scale, taking highest losses as most important to avoid (and breaking ties by moving to the next lower loss category). Similar as a standard equilibrium, a lexicographic equilibrium penalizes unilateral deviations, but in doing so, opens up a possibility for the second player to improve its own revenue in a (less important) other goal (e.g., causing more likely damage of a lower loss category). The theoretical facts about (real-valued) Nash equilibria in games, by the transfer principle, translate likewise to hyperreal terms. The practical difference relates to computability, since the defenses that we can find (algorithmically) in games over loss distributions are obtained from lexicographic equilibria. We will disambiguate the two hereafter by speaking about (standard) Nash equilibria to mean the classical concept, and lexicographic (Nash) equilibria to denote the other.

As for standard games, it can be shown that the saddle-point value $V(\vec
A)=\max_{\vec p\in \S(AS_1)}\min_{\vec q\in\S(AS_2)}F(\vec p,\vec q)$ is
invariant w.r.t. different (standard) Nash equilibria, and that equilibria defined w.r.t.
$\preceq$ exist (and can be generalized to standard Nash-equilibria in $n$-person
games in the canonic way). The way of proving it makes use of the embedding
of distributions into the hyperreal space $\HR$, where all the known
results necessary to re-establish the fundament of game theory are
available (yet further substantiating our loss representation by a moment
sequence). Unfortunately, however, not all properties are directly
inherited, such as a central computational feature of zero-sum games is
absent in our setting:

\begin{prop}\label{prop:fp-failure}
There exist zero-sum matrix games $\vec A\in\F^{n\times m}$ for which
fictitious play (according to \cite{Berger2007,Robinson1951}) does not
converge.
\end{prop}

Proposition \ref{prop:fp-failure} is proved by constructing a concrete
example (see \cite{Rass2015g}) of a game that cannot be solved using
fictitious play. Thus, it is an unfortunate obstacle in applying well-known
game theory to our new setting. The formal fix relies on yet another
representation of the loss densities, which admits converting a matrix game
over $\F$ into a set of standard matrix games over $\R$, which can be
solved by fictitious play again. However, the result is still not just a Nash equilibrium over the hyperreal numbers, as was observed by \cite{burgin_remarks_2021}, but has a lexicographic optimality property that is nonetheless appropriate for our matters of risk management. The details of this are postponed until
Section \ref{sec:practicalities}, culminating in the main Theorem
\ref{thm:main} that assures that we can ultimately escape the situation
that proposition \ref{prop:fp-failure} warns us about. Interestingly, this
fix has a useful side-effect, whose physical meaning is a heuristic account
for zero-day exploits. We will revisit this aspect in more detail later in
Section \ref{sec:zero-day}.

\section{APTs as Games}\label{sec:apt-games}
Suppose that a \ac{TVA} has been done and that an attack graph is
available. Towards a game-theoretic model of \acp{APT}, let us think of the
attack graph as sort of an \acf{EFG} with perfect information. Although the
attack graph or tree may not follow the proper syntax of an \ac{EFG}, we
can nevertheless convert it into syntactically correct normal form game in
the same way as we would do with an \ac{EFG}. That is, we would traverse
the graph from the initial stage of the game until the stage where payoffs
are issued to all players. From the exhaustive list of all these paths, we
can define the strategies of both players as rules about what to do at
which stage, given the other player's move. Likewise, an \ac{APT} would in
this view be mounted along any of the existing paths from the root of the
attack graph down to the goal, with the difference to \ac{EFG} mostly being
the fact that the ``game'' does not clearly define when the players are
taking their moves (this is a conceptual difference to \ac{EFG}, where the
assignment of which player's move it is part of the \ac{EFG} description).

In both cases, \ac{EFG} and \ac{APT} attack graphs, we can compile a set of
paths from the start to the finish, from which strategies for both players
can be identified. While this identification comes from the definition of the
\ac{EFG}, for \acp{APT}, the strategies are delivered only for the opponent
player 2, which is the attacker. Player 1, the defender, needs to derive its
action set $AS_1$ based on player 2's actions $AS_2$. Table
\ref{tbl:game-correspondence} summarizes the correspondence between \ac{EFG}
and \ac{APT} attack trees.

\begin{table}[h!]
  \caption{\textbf{Correspondence of Attack Trees/Graphs and Extensive Form Games}}\label{tbl:game-correspondence}
  \begin{tabular}{|p{0.4\columnwidth}|p{0.4\columnwidth}|}
  \hline
  \textbf{Extensive form game} & \textbf{Attack tree/graph} \\\hline
  start of the game & root of the tree/graph \\\hline
  stage of the gameplay & node in the tree/graph \\\hline
  allowed moves at each stage (for the adversary) & possible exploits at each node \\\hline
  end of the game & leaf node (attack target)\\\hline
  strategies & paths from the root to the leaf \mbox{($=$ attack vectors)}\\\hline
  information sets & uncertainty in the attacker's current position and move\\\hline
\end{tabular}
\end{table}

The nature of \acp{APT} induces a difficulty in the game specification here,
since we usually do not know how deep the attacker may have penetrated into
the system, and because of this, the current stage of the game is expectedly
unknown to the defender. Countermeasures against exploits in each stage may
be identified, but not always possible, feasible or successful. Allowing for
a random outcome with the possible event of an action to fail elegantly
tackles this issue in our setting.

If countermeasures have no permanent effect or are likely to fail,
then we may need to repeat them. For example, a security training may cause
only temporarily raised security awareness. Likewise, updating a software
once is clearly useless unless the system is continuously kept up to date.

Given that the defense actions in $AS_1$ must be repeated, we can set up a
matrix game to tell us the best way to do so. Since precautions cannot be
applied everywhere at all times, we need to define the game as one where
the defender takes random moves, based on a hypothesis where the attacker
may currently be. Alas, it would probably not be feasible to rely on
Bayesian updating towards refining our hypotheses, since this assumes much
data, i.e., many incidents to occur, and hence is exactly what \ac{APT}
mitigation seeks to prevent. Thus, many popular tools from game theory like
perfect Bayesian equilibria or related appear unattractive in our setting.

To simplify the issue in general, let us assume that the defender can access
all parts of the system (thus, take moves related to any stage of the game as
defined by the attack tree), whereas the attacker can move only from its
current location (node) to the successor (child node) location. The game play
is thus a matter of the defender seeking the optimal way of applying its
threat mitigation moves anywhere at random in the infrastructure, in an
attempt to keep the adversary away from its target. Unlike a Bayesian or
sequential game approach, the application of defense actions is not based on
hypotheses of where the attacker currently is, but will assume a worst-case
behavior of the attacker (thus, switching from the Bayesian towards the
minimax decision theoretic paradigm).

For example, applying a patch at some point may close a previously
established backdoor and send the adversary back to the start again.
However, if the patch is currently unavailable, not effective or simply
applied to the wrong machine, the defense move will have no effect at all.
Towards modeling this uncertainty, let us first become more specific on
what the payoffs in the example game will be.

Following the common qualitative risk assessments, we may define categories
of risk depending on ``how far away'' the adversary is from its destination
in the attack graph. Collecting the lengths of all paths listed in Table
\ref{tbl:as2}, we see that their lengths range between 4 and 8 nodes
(including the start \texttt{execute(0)} and finish node
\texttt{full\_access(2)}). In any such case, we may simply map the
distances to qualitative scores, such as Table \ref{tbl:distance-mapping}
proposes here.

\begin{table}[h!]
  \caption{\textbf{Possible mapping of graph distance to risk categories}}\label{tbl:distance-mapping}
  \begin{tabular}{|l|l|}
    \hline
    \textbf{Distance} & \textbf{Risk} \\\hline
    7\ldots 8 & low \\\hline
    3\ldots 6 & medium \\\hline
    0\ldots 2 & high \\
    \hline
  \end{tabular}
\end{table}

The concrete mapping of distances to risk levels can, however, already
induces uncertainty. For example, assume that an attacker has already
gained execution privileges on machine 1 (denoted as \texttt{execute(1)} in
the attack graph), then it may either continue its way on path 4 in Table
\ref{tbl:as2} (via a remote \ac{FTP} connection from machine 1 to machine
2; node \texttt{ftp\_rhosts(1,2)}) or on path 5 in Table \ref{tbl:as2} (via
an \ac{RSH} connection from machine 1 to machine 2; node
\texttt{rsh(1,2)}). On path 4, the distance to \texttt{full\_access(2)} is
3 nodes (\texttt{ftp\_rhosts(1,2)} $\to$
\texttt{rsh(1,2)} $\to$ \texttt{local\_bof(2)}), while on path 5, the
distance is only 2 nodes (\texttt{rsh(1,2)} $\to$ \texttt{local\_bof(2)}).
In light of the a-priori specified mapping of distance to risk levels as in
Table \ref{tbl:distance-mapping}, the risk would be classified as
\emph{either} ``medium'' or ``high'', depending on which path has been
chosen. The usual stealthiness of \acp{APT} hence causes uncertainty in the
risk assessment, which needs to be captured by a proper decision- or
game-theoretic \ac{APT} mitigation approach.

\subsection{Identifying Mitigation Strategies}

Having the attack graph and once attack vectors have been derived from it,
all of which are collected in the adversary's action space $AS_2$, the next
step is the identification of mitigation strategies. This process is a
standard phase in many risk management practices, and often based on known
countermeasures against the identified threats (accounting for unexpected
events is a matter of zero-day exploit handling, which we will revisit
shortly in Section \ref{sec:zero-day}). Since the process of defining
countermeasures is a task that highly depends on the attack vectors $AS_2$,
we cannot define a general purpose procedure to identify $AS_1$ here (it is
individual and different for various infrastructures). For the sake of
generality and conciseness of presentation in this work, let us therefore
assume that all relevant defense actions are available and constitute the
action set $AS_1$ for the defender. It may well be the case that not all
actions are effective against all threats, and neither may a designated
countermeasure $i$ be necessarily effective against threat $j$. In that
case, we may pessimistically assume maximal damage to be likely in scenario
$(i,j)$. Matching all defenses in $AS_1$ against all attack vectors in
$AS_2$ is a matter of defining the game's loss distributions, which is the
next step towards completing the game model in Section
\ref{sec:defining-the-apt-game}. To simplify the notation in the following,
let us abstractly denote the action spaces as $AS_1=\set{1,2,\ldots,n}$ and
$AS_2=\set{1,2,\ldots,m}$, with the specific details of the $i$-th defense
and the $j$-th attack for all $i,j$ being available in the risk management
documentation (in the background of our modeling).

\subsection{Defining the \ac{APT} Game}\label{sec:defining-the-apt-game}

Towards a matrix game model of \acp{APT}, it remains to specify the outcome
of each attack/defense scenario. To capture the intrinsic uncertainty here,
we will resort to a qualitative assessment like sketched above and/or
arising from vague opinions like Table \ref{tbl:security-precaution}
illustrates. The \ac{APT} game is then defined upon loss distributions
according to definition \ref{def:loss} to describe the potential loss in
each scenario $(i,j)\in AS_1\times AS_2$. In cases where we are unable to
come up with a reasonable guess on the distributions, a pessimistic
approach towards a worst-case assessment could work as follows (we will
later revisit the issue in the discussion Section \ref{sec:faqs}):
\begin{enumerate}
  \item Fix a particular position in the network, i.e., a certain point
      where an attack is considered. Let the hypothesized attack be the one
      with index $j_0\in AS_2$.
  \item Fix a defense action $i\in AS_1$.
  \item In lack of better knowledge, assume a uniform distribution of all
      possible $j$, including $j_0$, and rate the success probability
      $p_{j_0}$ of the defense conditional on $j=j_0$ (i.e., if your
      guess was right). This rating can also be made conditional on the
      expected ``doability'' of the current defense (e.g., if the defense
      means patching, the patch may not be available at all times, or it
      may be ineffective). Fig \ref{fig:loss-assessment} displays the
      process as a decision tree, in which the worst case outcome
      (highlighted in gray) is taken if either the countermeasure is
      considered as possibly effective but may still fail (with a certain
      likelihood), or if the countermeasure is not applicable at all. In
      any case, the expert -- based on the assumed attacker's behavior --
      is not bound to confine her/himself to a single answer, and may
      rate all the possibilities with different likelihoods
  \item If possible, collect many such assessments (say, from surveys,
      simulations, etc.), and compile an empirical distribution from the
      available data. This empirical distribution is then nothing else
      than a histogram recording the number of uttered opinions (shown as
      the bar chart in Fig \ref{fig:loss-assessment}) Note that the
      uniformity assumption on the attacker's location can be replaced by
      a more informed guess, if one is available. For example, the
      adversarial risk analysis (ARA) framework
      \cite{RiosInsua2009,Rios2012,Rothschild2012} addresses exactly this
      issue.
\end{enumerate}
This procedure is repeated for the entirety of scenarios in $AS_1\times
AS_2$, i.e., until the full game matrix has been specified. Fig
\ref{fig:workflow} illustrates the three sub-steps per entry of this
procedure in the game matrix (note that the right-most decision path
(bold-printed in Fig \ref{fig:loss-assessment}) is reflected as the fourth
choice option in the example survey shown in Fig \ref{fig:workflow}). We
stress that the full lot of $\abs{AS_1}\cdot\abs{AS_2}$ may hardly be
necessary to put to a survey, since not all actions are effective against
one another (defenses may work against specific threats, so only a few
combinations in $\abs{AS_1}\cdot\abs{AS_2}$ need to be polled explicitly,
and others can rely on default settings; cf. Fig
\ref{fig:loss-assessment}). Furthermore, some loss distributions can be
equally well computed from simulations (cf., e.g., \cite{Koenig2016a}).

\begin{figure}
  \includegraphics[scale=1]{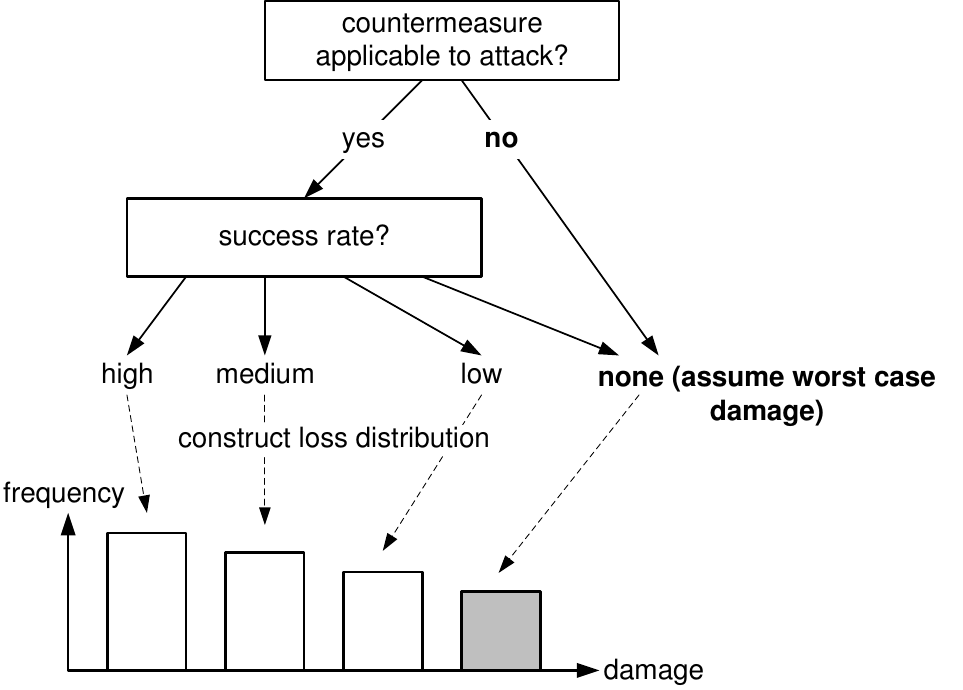}
  \caption{\textbf{Loss Assessment of Counteraction vs. Threat}}\label{fig:loss-assessment}
\end{figure}

\begin{figure*}
  \includegraphics[scale=0.6]{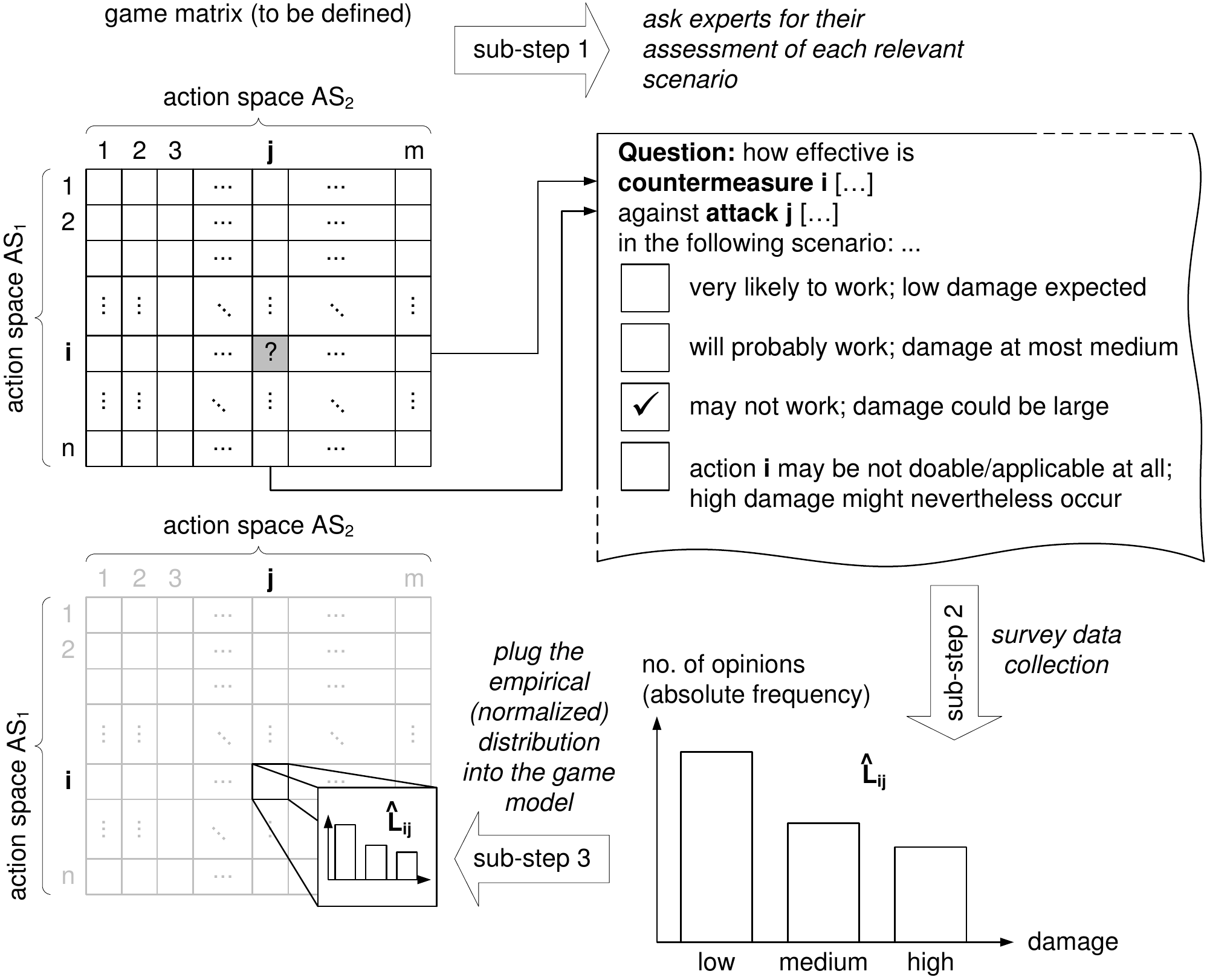}
  \caption{\textbf{Specification of an \ac{APT} Game (Example Workflow Snapshot)}}\label{fig:workflow}
\end{figure*}

By construction of the total ordering, the game that we define to minimize
the loss would then be played towards minimizing the likelihood for large
damages (by Theorem \ref{thm:preference-meaning}). Returning to our example
sketched in Section \ref{sec:running-example}, the uncertainty in a risk
level quantification based on distance in the graph would thus mean that
the gameplay is such as to keep the adversary ``as far away as possible''
from its target. This is indeed what we would naturally expect, and the
$\preceq$-relation acting on loss distributions that are based on distance
achieve precisely this kind of defense.

Although this modeling of \acp{APT} is heavily based on (subjective)
expertise and manual labour, it fits quite well into standard risk
management processes (such as ISO 31000 \cite{ISO31000} or ISO 27005
\cite{ISO27005}), and nevertheless greatly simplifies matters of modelling
over the classical approach, as a variety of issues are elegantly solved. A
selection is summarized in Table \ref{tbl:benefits}, with a complementary
discussion given in Section \ref{sec:faqs}. Additional help in the
specification of risk assessments is also offered by thinking about costs
of an exploit or known ratings of vulnerabilities such as by \ac{CVSS}
\cite{Mell2007}. Such ratings are commonly delivered along with the
\ac{TVA} (e.g., by tools like OpenVAS).

\begin{table*}[h!]
  \caption{\textbf{Benefits of Distribution-Valued Game-Modeling over Classical Game-Modeling}}\label{tbl:benefits}
  \begin{tabular}{|p{0.1\textwidth}|p{0.4\textwidth}|p{0.4\textwidth}|}
    \hline
    \textbf{Issue} & \textbf{Classical game-theoretic modelling} & \textbf{How this is handled in distribution-valued games} \\\hline
    Payoff uncertainty & Either switching to special forms of equilibria (disturbed, trembling hands, etc.)
    or agreeing on a simultaneously representative value for all possible outcomes (``consolidation of different opinions'') & No consolidation or representation needed;
    we can simply work with the (normalized) histogram of all possible outcomes (or opinions on what could happen)\\\hline
    Non-realizable strategy & Separating out cases where a strategy can be played or not. This would amount to specifying two versions
    of the strategy (one that is successful and one that fails) & Since actions can by construction have many different outcomes, success and failure are just two
    realizations of the corresponding loss \ac{RV} $L$, each of which may occur with a known (or estimated) probability. The entirety of these probabilities makes the
    sought density function $f$ of the \ac{RV} $L$.\\\hline
    Imperfect information & Working with hypotheses on expected moves in stages of the game where no precise information is available.
    The hypotheses can be learnt from past history and are taken into account when defining the optimal
    behavior (e.g., Bayesian perfect equilibrium) & Is directly incorporated in the uncertainty of the outcome, since an unknown move corresponds in a perceived random payoff; thus, there is
    no intrinsic conceptual difference here\\\hline
    Random changes in the game-play (stochastic games \cite{Shapley1953a}) &
    Resorting to special forms of equilibria, such as distorted or trembling hands equilibria or stochastic games \cite{Shapley1953a,Fudenberg1991} &
    As long as the outcome remains identically (stationarily) distributed across several rounds of the gameplay, there
    is no specific treatment required upon random changes in the gameplay. The known theory of Markov chains can be used here
    to analyze the changes in the gameplay for stationarity.\\
    \hline
  \end{tabular}
\end{table*}

\section{Practical Computation of Optimal Defenses}\label{sec:practicalities}
Essentially, our \ac{APT} game model is a matrix game $\vec
A\in\F^{\abs{AS_1}\times\abs{AS_2}}$, in which each defense $i\in AS_1$ vs.
each attack in $j\in AS_2$ is rated in terms of a probability distribution
(uncertain outcome) $F_{ij}\in\F$. By defining the losses in the gameplay
to be the gain for the adversary (i.e., making the competition zero-sum),
we obtain a valid worst-case approximation that enjoys the following useful
property:

\begin{lem}\label{lem:zero-sum-optimality}
Let $AS_1, AS_2$ be the action spaces for the defender and the attacker,
respectively, with cardinalities $n$ and $m$. Furthermore, let $\vec B$ be
the (unknown) matrix of true payoffs for the attacker, and let $\vec A$ be
the loss matrix for the defender. If the saddle-point values of the zero-sum
matrix game $\vec A$ is $V(\vec A)$, and $V(\vec A,\vec B)$ is any standard
equilibrium payoff in the bimatrix game induced by $\vec A, \vec B$, then we
have:
\begin{equation}\label{eqn:zero-sum-bound}
    V(\vec A,\vec B)\preceq V(\vec A),
\end{equation}
provided that the defender plays a zero-sum equilibrium strategy (induced by
$\vec A$) in both games, the zero-sum game $\vec A$ and the bi-matrix game
$(\vec A,\vec B)$.
\end{lem}
Lemma \ref{lem:zero-sum-optimality} directly follows from the definition of
standard Nash equilibria, and is a well known fact; cf. \cite{Alpcan2010} for a more
elaborate discussion. The computation of equilibria in the sense of Lemma \ref{lem:zero-sum-optimality} requires hyperreal arithmetic, but lexicographic Nash equilibria are computable by conventional means only, and the bound in \eqref{eqn:zero-sum-bound} then remains valid w.r.t. a descending order of categories on the loss scale. This is nothing but a risk-averse optimization of worst-case outcomes.

Intuitively, Lemma \ref{eqn:zero-sum-bound} says that the worst case attack
occurs when the adversary's incentives are exactly opposite to our own
interest. In particular, observe that the upper bound
\eqref{eqn:zero-sum-bound} is independent of the adversary's
payoff/incentive structure $\vec B$, and optimizes the chances to suffer worst-case losses (and breaking ties by optimizing the losses in descending order of categories). The irrelevance of $\vec B$ in the
upper bound tells that we do not require any information about the
adversary's intentions or incentives, as $V(\vec A)$ can be determined only
based on the defender's possible losses, and for a lexicographic order of losses, corresponding to a worst-case avoiding defense. Thus, in playing the zero-sum
defense we obtain a baseline security that is guaranteed irrespectively of
how the adversary behaves, provided it acts only within its action set
$AS_2$. The case of unexpected behavior, that is, actions outside $AS_2$,
corresponds to an unforeseen zero-day exploit. A remarkable feature of the
game model using distribution-valued payoffs is its natural account for
such outcomes, which we will further discuss in Section \ref{sec:zero-day}.

Our use of non-standard calculus somewhat limits the practically doable
arithmetic, since for example, divisions in $\HR$ require an ultrafilter to
fully describe $\HR$ (which is unavailable since $\HR$ is defined
non-constructively; see \cite{Robinson1966}). Fortunately, however, these
issues do not apply for our matrix games here, since lexicographic Nash equilibria can
still be computed by fictitious play (FP)
\cite{Robinson1951,Berger2007,Washburn2001} which uses practically doable
$\preceq$-comparisons only. The unpleasant possibility of non-convergent FP
(proposition \ref{prop:fp-failure} warns us about this) is escaped by using
non-parametric (kernel density) models for the payoff distributions, and
using Lemma \ref{lem:derivative-preferences} to decide $\preceq$. The
respective details are laid out in section \ref{sec:derivative-decisions}.
This closes a gap left open in \cite{Rass2015d}.

\subsection{Computing Optimal Defenses (Equilibria)}\label{sec:derivative-decisions}
To avoid proposition \ref{prop:fp-failure} to apply for our \ac{APT}-games,
we need to assure that all loss distributions share the same support (the
counterexample used to prove proposition \ref{prop:fp-failure} relies on
different losses whose supports that are strictly contained in one
another). To this end, we will introduce another representation of a loss
distribution as a sequence, so that the lexicographic ordering $<_{lex}$ on
the new sequence equals the $\preceq$-ordering of definition
\ref{def:preference}.

To make things precise, let us assume that a particular empirical loss
distribution $\hat L_{ij}$ has been compiled from the data
$x_1,x_2,\ldots,x_N$, as obtained from simulations or expert
questionnaires. Moreover, assume that $\hat L_{ij}$ is categorical, so that
the underlying data points (answers) can be ranked within a finite range
(in the example in Fig \ref{fig:workflow}, we have three categories, e.g.,
\{``low'', ``medium'', ``high''\}, which correspond to the ordered ranks
$\set{1,2,3}$). The case when the loss distributions is continuous is in
fact even simpler, and discussed later in remark
\ref{rem:continuous-losses}. For now, let us stick with the expectedly more
common practical case where risk assessments are made in categories rather
than hard figures.

First, the empirical distribution (normalized histogram) $\hat L_{ij}$ is
replaced by a \ac{KDE} using Gaussian kernels of a fixed bandwidth to
define
\begin{equation}\label{eqn:kernel-approximation}
\tilde f_{\hat L_{ij}}(x) = \frac 1{N\cdot h}\sum_{k=1}^N K\left(\frac{x_k - x}h\right),
\end{equation}
where
\begin{itemize}
  \item $K(x)=\frac 1{\sqrt{2\pi}}\exp\left(-\frac 1 2 x^2\right)$ is the
      standard Gaussian function,
  \item and $h>0$ is a bandwidth parameter that can be estimated using
      (any) standard statistical rule (of thumb, e.g., Silverman's
      formula \cite{Silverman1998}).
\end{itemize}
Although any such nonparametric estimation can be quite inaccurate, yet as
the data on which it is based is subjective anyway, the additional
approximation error may not be as significant (nor in any sense
quantifiable; still, Nadaraja's theorem (see \cite{Wand1995}) would assure
that a continuous unknown underlying loss distribution would be
approximated arbitrarily well in probability, as $N\to\infty$, provided
that $h$ is chosen as $h(N)=c\cdot N^{-\alpha}$ for any two constants $c>0$
and $0<\alpha<1/2$.)

Using the \ac{KDE} \eqref{eqn:kernel-approximation}, we can cast the
distribution-valued game back into a regular matrix-game over the reals.
Note that in choosing Gaussian kernels, we naturally extend all density
functions to the entirety of $\R$. Such distributions would not be losses
in the sense of definition \ref{def:loss}, as the supports are no longer
bounded.

Using the aforementioned risk acceptance threshold $a>1$ to truncate all
loss distributions at $a$, casts the loss distributions into the proper
form. We can then expand the (truncated) loss density $f_{\hat L_{ij}}$
into a Taylor series for every scenario $(i,j)\in AS_1\times AS_2$. To ease
notation in the following, let us drop the double index and simply write
$\tilde{f}$ to mean the kernel density approximation of the empirical
distribution in the given scenario, based on $N$ data samples. Then, its
Taylor series expansion at point $a$ is
\begin{equation}\label{eqn:kde-derivative}
    \tilde f(x) =\sum_{k=0}^\infty \frac{(x-a)^k}{k!}\tilde f^{(k)}(a),
\end{equation}
which converges everywhere on $[1,a]$ by our choice of the Gaussian kernel.
The $k$-th inner derivative is obtained from the kernel density definition
\eqref{eqn:kernel-approximation} as,
\begin{equation}\label{eqn:kernel-derivative}
    f^{(k)}(x)=\frac 1{N\cdot h^2\cdot \sqrt{2\pi}}\sum_{j=1}^N \frac{d^k}{dx^k}\exp\left(-\frac 1{2h^2}(x_j - x)^2\right)
\end{equation}
Here, our use of Gaussian density pays a second time, since the $k$-th
derivative of the exponential term can be expressed in closed form using
Hermite polynomials by exploiting the relation
\begin{equation}\label{eqn:gauss-derivatives}
    (-1)^k\exp\left(\frac{x^2}2\right)\frac{d^k}{dx^k}\exp\left(-\frac{x^2}2\right)=2^{-\frac k 2}H_k\left(\frac x{\sqrt 2}\right),
\end{equation}
in which $H_k(x)$ is the $k$-th Hermite polynomial, defined recursively as
$H_{k+1}(x) := 2xH_k(x)-2 H_{k-1}(x)$ upon $H_0(x)=1$ and $H_1(x)=2x$.
Plugging \eqref{eqn:gauss-derivatives} into \eqref{eqn:kernel-derivative},
and after rearranging terms, we find

\begin{equation}\label{eqn:density-closed-form}
f^{(k)}(x) =
\frac 1{N\sqrt{\pi}}\frac{(-1)^k}{(h\cdot\sqrt{2})^{k+1}}
\times\sum_{j=1}^n\left[H_k\left(\frac{x-x_j}{h\sqrt 
2}\right)\cdot\exp\left(-\frac{(x-x_j)^2}{2h^2}\right)\right]
\end{equation} 

Evaluating the derivatives up to some order and substituting the values
back into \eqref{eqn:kde-derivative}, we could numerically construct the
kernel density estimator. Fortunately, there is a shortcut here to avoid
this, if we use the vector of derivatives with alternating signs to
represent the Taylor-series expansion, and in turn the \ac{KDE}, by
\begin{equation}\label{eqn:series-representation}
    \tilde f_{\hat L_{ij}} \simeq \left((-1)^k f^{(k)}_{\hat L_{ij}}(a)\right)_{k=0}^\infty = (y_0, y_1, y_2, \ldots)\in\R^{\infty},
\end{equation}
where the entries of the sequence can be computed from
\eqref{eqn:density-closed-form}.

Interestingly, under the assumptions made (i.e., truncation at a point
$1<a\in\R$ and approximating the empirical distribution by a Gaussian
\ac{KDE}), the \emph{lexicographic order} on the series representation
\eqref{eqn:density-closed-form} \emph{equals} the preference order
$\preceq$ on the hyperreal representation of the loss distribution (see
Lemma \ref{lem:derivative-preferences} in the appendix for a proof). That
is, we can decide $\preceq$ between two sequence representations
$(y_k)_{k=0}^{\infty}$ and $(z_k)_{k=0}^{\infty}$ of the form
\eqref{eqn:series-representation} as follows:
\begin{itemize}
  \item If $y_0<z_0$, then $L_1\preceq L_2$. If $y_0 > z_0$, then
      $L_2\preceq L_1$. Otherwise, $y_0=z_0$, and we check
  \item if $y_1<z_1$, then $L_1\preceq L_2$. If $y_1>z_1$, then $L_2\preceq
      L_1$. Otherwise, $y_1 = z_1$, and we check
  \item if $y_2<z_2$, then $L_1\preceq L_2$, etc.
\end{itemize}

\begin{rem}[Continuous loss
models]\label{rem:continuous-losses} If a continuous loss model is
specified, differentiability may be not be an issue if the density $f$ has
derivatives of all orders. Otherwise, we can convolve $f$ by a Gaussian
density $k_h$ with small variance $h>0$ to get an approximation $\tilde
f=f\ast k_h\in C^\infty$ at any desired precision. Kernel density estimates
are exactly such convolutions and thus provide convenient differentiability
properties here.
\end{rem}

Experimentally, we observed that in many cases the preference decision
$\preceq$ can be made already using the first value $f(a)$ in the sequence
\eqref{eqn:series-representation}. If the decision cannot be made (upon a
tie $f_{L_1}(a)=f_{L_2}(a)$), then we can move on to the first order
derivative, and so on. Thus, we technically do fictitious play in parallel
on a ``stack'' of matrix games $\vec A_0, \vec A_1, \vec A_2,
\ldots\in\R^{\abs{AS_1}\times \abs{AS_2}}$, where the $k$-th matrix is
constructed (only on demand) with the $k$-th entry of the sequence
representation \eqref{eqn:series-representation}. The selection of
strategies is herein always made on the first game matrix $\vec A_0$,
looking at the others only in cases where the decision cannot be made
directly within $\vec A_0$ (see Fig \ref{fig:game-stack} for an
illustration). Since we are now back at a regular matrix game, the usual
convergence properties of fictitious play are restored.

Of course, we cannot run fictitious play on an infinite stack of matrix
games, so we are necessarily forced to restrict attention to a finite
``sub-stack''. However, depending on how ``deep'' the stack is made, we can
reach a lexicographic equilibrium at arbitrary precision. This is made rigorous in the
following definition:

\begin{defn}[Approximate Equilibrium]\label{def:approximate-equilibrium}
Let $\eps>0,\delta>0$ be given, and let $\vec A=\F^{n\times m}$ be a
zero-sum matrix game with distribution-valued payoffs. We call a strategy
profile $(\tilde{\vec p}^*,\tilde{\vec q}^*)\in\R^{n+m}$ an
$(\eps,\delta)$-approximate equilibrium, if there is an equilibrium $(\vec
p^*,\vec q^*)\in\R^{n+m}$ (in the zero-sum game $\vec A$) such that both of
the following conditions hold:
\begin{enumerate}
  \item $\norm{(\tilde{\vec p}^*,\tilde{\vec q}^*)-(\vec p^*,\vec
      q^*)}_{\infty}<\eps$, and
  \item $\norm{\tilde{F}^* - F^*}_{L^1}<\delta$,
\end{enumerate}
where the equilibrium payoffs $\tilde{F}^*$ and $F^*$ are defined by
\eqref{eqn:utility-function} upon the equilibrium and its approximation. 
\end{defn}

\begin{rem} By the equivalence
of norms on $\R$, the $\infty$-norm in condition 1 can be replaced by any
other norm upon, as long as the value of $\eps$ is chanced accordingly.
\end{rem}
\begin{rem}
The continuous dependence of $F(\vec p,\vec q)$ on $(\vec p,\vec q)$ (in
      terms of the topologies on $\R^{n+m}$ and $L^1$; the latter space
      being the one where the payoff distributions live in) as implied by
      \eqref{eqn:utility-function} leads to thinking that letting $\eps\to
      0$ would also cause $\delta\to 0$. This impression is not necessarily
      true, since the $\eps$-condition in definition
      \ref{def:approximate-equilibrium} can be taken as a convergence
      criterion when iteratively computing equilibria, while the payoff
      distributions can be approximated at a precision that is independent
      of this value $\eps$. Indeed, the goodness of approximation is
      controlled by the amount of data and the parameter $h$ chosen for the
      \ac{KDE} or the mollifier (cf. remark \ref{rem:continuous-losses}),
      and as such is independent of $\eps$. Hence, asking for a
      $\delta$-deviation for the payoff distribution accounts for a
      nontrivial degree of freedom here.
\end{rem}

While existence of equilibria in mixed strategies for all finite games
$\vec A\in\F^{n\times m}$ is assured (see \cite{Rass2015d}), the existence
of an approximate equilibrium is not as obvious. In fact, the actual use of
an approximate equilibrium is to find it within the set of regular matrix
games, so that games with payoffs from $\F$ can be solved for equilibria
just like a standard matrix game would be treated. Clearly, since
$\R\subset\HR$ and by virtue of the embedding $\phi$, the set of matrix
games $\vec A\in\R^{n\times m}$ forms a subclass of games of the form $\vec
A\in\HR^{n\times m}$, which itself covers games of the form $\vec
A\in\F^{n\times m}$ by virtue of $\phi$. A practical method to solve games
over $\F$ is obtained by approximating the solution from the inner set of
equilibria in matrix games with real-valued payoffs. This is the main
theorem of this work, whose proof is delegated to the appendix.

\begin{thm}[Approximation Theorem]\label{thm:main}
For every $\eps>0, \delta>0$ and every zero-sum matrix game $\Gamma_1 =
\vec A\in\F^{n\times m}$ with distribution-valued payoffs, there is another
zero-sum matrix game $\Gamma_2=\vec B\in\R^{n\times m}$ so that an
equilibrium in $\Gamma_2$ is an $(\eps,\delta)$-approximate lexicographic equilibrium in
$\Gamma_1$.
\end{thm}

\begin{figure}
  \includegraphics[scale=0.7]{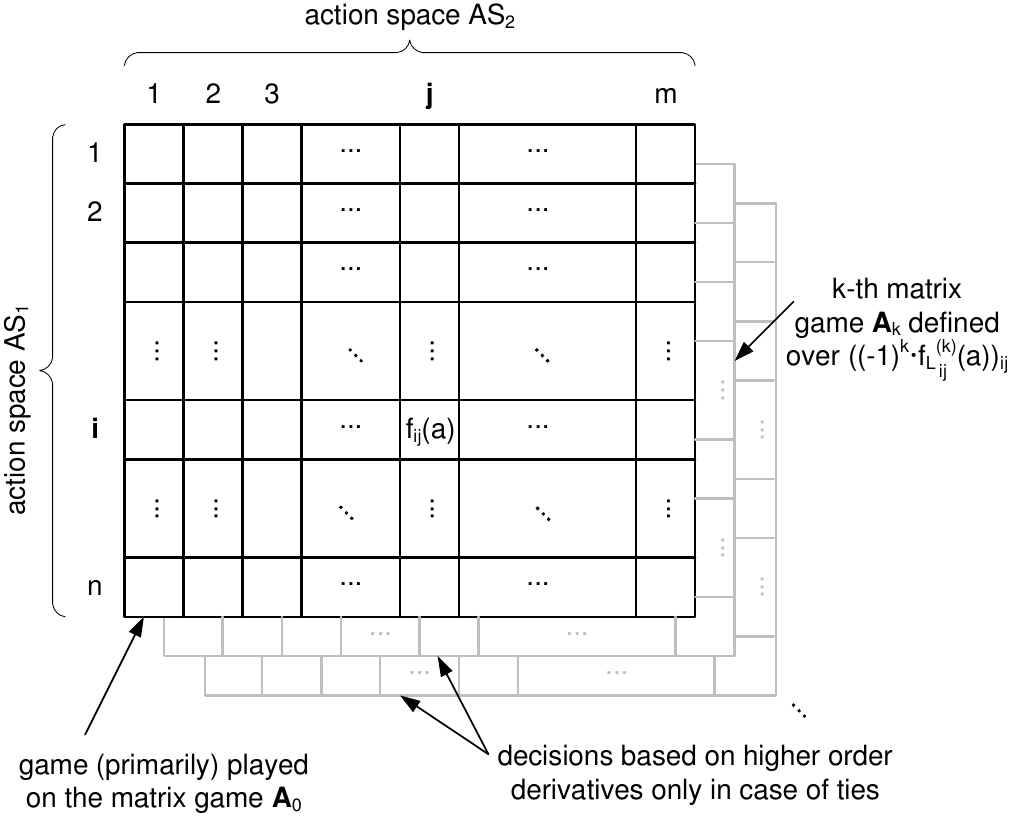}
  \caption{\textbf{Applying Fictitious Play}}\label{fig:game-stack}
\end{figure}

\subsection{Zero-Day Exploits}\label{sec:zero-day}
As a matter of consequence from the \ac{KDE} approximation of empirical
densities using Gaussian kernels, the approximate density in any case is
supported on the entire real line. That is, the density function
\eqref{eqn:kernel-approximation} assigns positive likelihood to the entire
range $(a,\infty)$, where $a$ is again our risk acceptance threshold. For the
specific scenario $(i,j)$, this also means that positive likelihood is
assigned to losses in the range $Z=(\max_k\set{x_k},\infty)$, where
$x_1,\ldots,x_N$ are the observations upon which our empirical loss
distribution is based, and $Z$ is the range of losses that were never
observed. These are, by definition, exactly the events of zero-day exploits.
More importantly, losses in $Z$ have -- by construction -- a positive
likelihood to occur in scenario $(i,j)$ under the approximate \ac{RV}
$\tilde{L}_{ij}$ with density function \eqref{eqn:kernel-approximation}.

In other words, no matter of whether or not we explicitly sought to model
zero-day exploits, they are automatically (implicitly) taken into account
by the loss density approximation technique laid out in Section
\ref{sec:practicalities}. The specific likelihood for a zero-day exploit,
as based on the information available, is simply the mass assigned to $Z$
under $\tilde{f}_{\hat L_{ij}}$ (as defined in
\eqref{eqn:kernel-approximation}). This value increases, the more
observations on higher losses are available, say, if more experts expect
higher damages to occur. In that case, the kernel density estimate will put
more mass on this area, thus fattening the tails of the \ac{KDE}
approximation \eqref{eqn:kernel-approximation}. Fig
\ref{fig:equilibrium-loss} will later display the lexicographic equilibrium outcome of an
example \ac{APT}-game model, showing the ``zero-day area'' in gray.

\begin{figure}[h!]
  \includegraphics[scale=1]{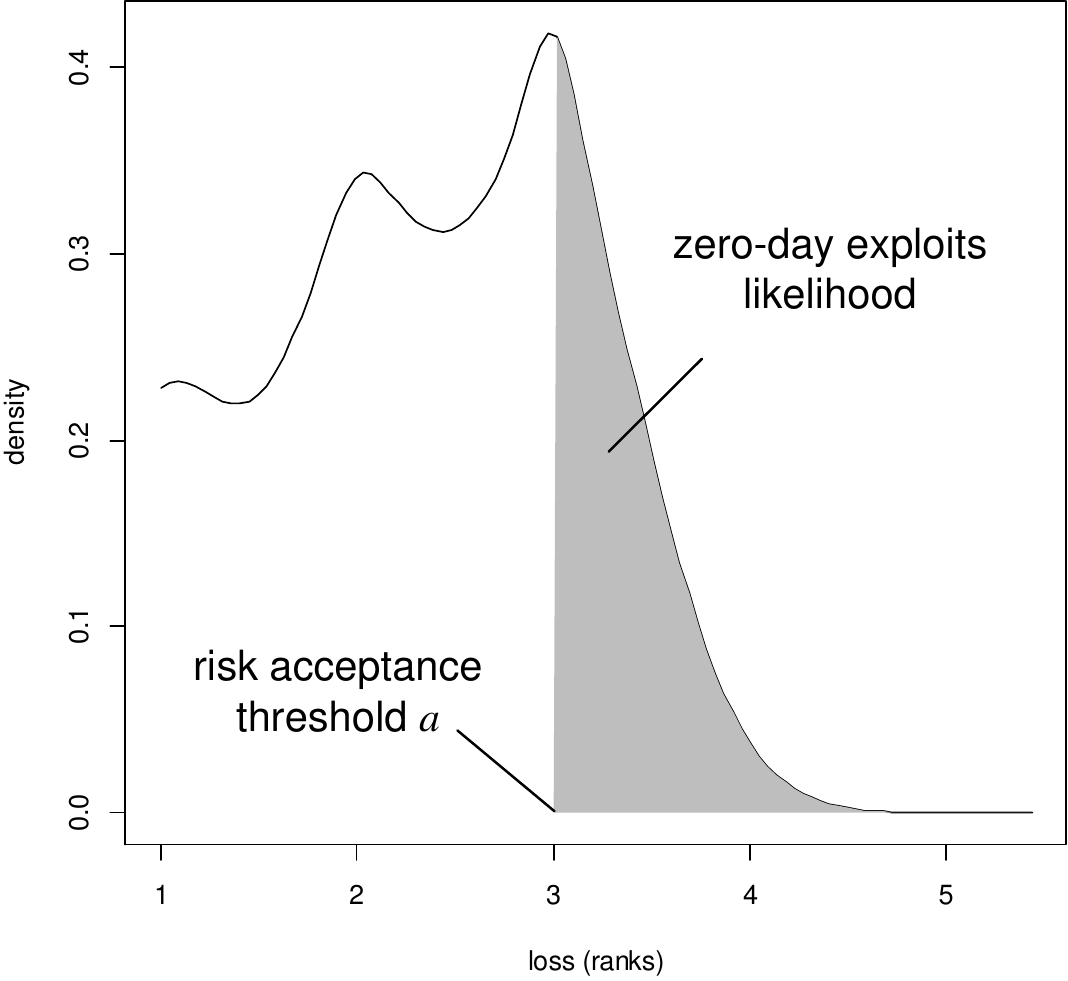}
  \caption{\textbf{Equilibrium loss distribution for the example \ac{APT} mitigation game}}\label{fig:equilibrium-loss}
\end{figure}

In connection with the design of our games to yield optimal behavior that
minimizes the chances for high losses, the likelihood for zero-day exploits
is then automatically minimized, since the $\preceq$-optimal decisions are
those that have all their masses shifted towards lowest damages as much as
possible. Therefore, practically, we can adjust our modeling account for
zero-day exploits by adding more pessimistic observations to the data sets
from which we construct the loss distributions. But from that point
onwards, the construction automatically considers extreme events in the way
as we want it without further explicit considerations.

\section{Generalizations and Special Cases}\label{sec:generalizations}

The \ac{APT} modeling can be generalized to deal with multiple relevant
interdependent aspects, such as different security goals (confidentiality
vs. integrity vs. availability, etc.) but also taking costs into account
(such as, for example, if the attack graph is enriched with information on
how much an exploit would cost the attacker, or the level of skills
required to mount an attack). By virtue of our embedding into the
hyperreals, all the known theory from multi-criteria game theory carries
over to this setting. Specifically, multi-criteria games as have been
studied by \cite{Voorneveld1999,Lozovanu2005,Rass2012,Rass2014} can be
analyzed in exactly the analogous way in our setting. This has been done in
\cite{Rass2015a}, so we confine ourselves to only sketching the approach
here.

\subsection{Multiple Goals and Optimal Tradeoffs}\label{sec:tradeoffs}
Suppose for simplicity that only a choice between finitely many options is
to be made, where risk is decreased upon increasing investments.
\begin{exa}[Uninterruptible power supplies]\label{exa:usv}
Imagine that a company ponders about installing additional power-supplies
to cover for outages. Depending on how many such systems are available (in
different subsidiaries of the company), the risk for an outage will
obviously decrease. Equally clear is the increase of costs per additional
system, so putting both of these outcome measures into a graph, we end up
with finding the optimal tradeoff somewhere in the middle. The question is
now how to find the optimum algorithmically.
\end{exa}

Formally and generally, let $F_{ij}^1, F_{ij}^2,\ldots, F_{ij}^d$ be
different measures of loss that all need to be accounted for in the
scenario $(i,j)\in AS_1\times AS_2$, and let
$\alpha_1,\alpha_2,\ldots,\alpha d\in (0,1)$ be weights assigned to these
losses (identically for all $i,j$). Practically, such losses could concern
(among others):
\begin{itemize}
  \item confidentiality of information (loss distributions $F_{ij}^1$),
  \item availability of systems (loss distributions $F_{ij}^2$),
  \item security investments/costs to run defense actions (loss
      distributions $F_{ij}^3$),
  \item etc.
\end{itemize}

Then, \cite{Lozovanu2005} has shown that a vector-valued loss referring to
multiple interdependent criteria can be converted into a simple
(single-criteria) loss by taking
\begin{equation}\label{eqn:convex-combination}
    F_{ij} := \alpha_1\cdot F^1_{ij} +\alpha_2\cdot F^2_{ij} + \ldots + \alpha_d\cdot F^d_{ij},
\end{equation}
into the optimization. Technically, the convex combination takes us to the
Pareto-front of admissible actions, and the resulting optima and
Nash-equilibria are understood in terms of Pareto-optimality. In
\eqref{eqn:convex-combination}, the addition is understood pointwise on the
distribution functions, which is mathematically justified since the
expectation operator is linear w.r.t. the distributions over which it is
computed (hence, the moment sequences arise in the proper form).
Consequently, multicriteria game theory as studied in
\cite{Voorneveld1999,Lozovanu2005,Rass2012} applies \emph{without change}
here.

The only technical constraint that applies here is that all $F_{ij}^k$ for
all $i,j,k$ must have the same support $\Omega\subset\R$. That is, we must
measure the loss in a common scale, for otherwise, the above scalarization
does not make sense.

\paragraph{Combining Losses of Different Nature}
Different goals may be measured in individual scales (such as monetary loss
being expressed as a number, loss of public reputation expressed in a
nominal scale like ``low/medium/high confidence'' or loss of customers
being an integer count). To harmonize these towards making the convex
combination \eqref{eqn:convex-combination} meaningful, we need to cast all
these scales into a common scale $\Omega$ over the reals. While there is no
general method to do this, practical heuristics to achieve this mainly do
two steps:
\begin{enumerate}
  \item Define an ordered set of fixed loss categories that shall apply
      for all goals of interest, and define the understanding of each
      category \emph{individually per goal}.
  \item Map all concrete (e.g., numeric) losses into the so-defined
      categories.
\end{enumerate}
This approach is, for example, followed in many national risk management
standards, such as
\cite{DeutscherBundestag2014,Winehav&Florin&Lindstedt2012,Hohl&Brem&Balmer2013}

Picking up example \ref{exa:usv}, let us assume that we seek to decide
between installing between zero and five auxiliary power supply systems.
Letting the priorities to be ``security:costs = 60:40'' (i.e.,
$\alpha_1=0.6,\alpha_2=0.4$), we find the optimal tradeoff to be between
two and three additional power supplies. Finding the actual optimum is then
a simple matter of comparing calculations for 2 and 3 power supplies since
all other options have been eliminated; see Fig \ref{fig:optimal-tradeoffs}
(the computation on loss distributions would be identical, only taking a
pointwise weighted sum of the distributions; only the resulting plot would
be much less illustrative than the shown picture of real-valued functions).

\begin{figure}
  \includegraphics[scale=0.7]{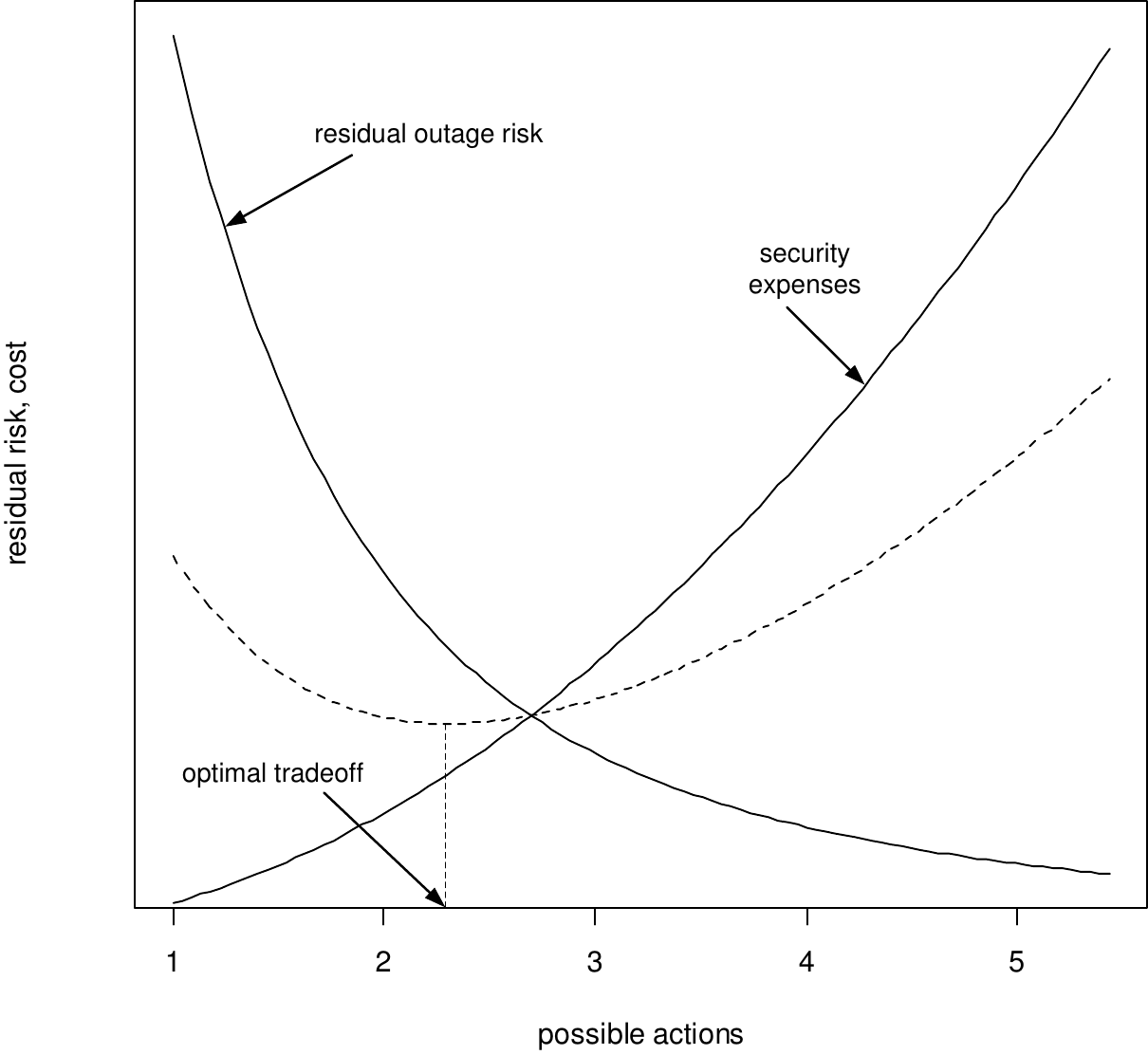}
  \caption{\textbf{Optimal Tradeoffs (simple case)}}\label{fig:optimal-tradeoffs}
\end{figure}

\subsection{Optimizing a (Permanent) Security Design}
Since the concept of Nash equilibrium on which optimal choices are based
assumes randomized (and thus repeated) actions, the question on how to apply
the optimization to ``permanent'' countermeasures (that are not repeated
actions in the classical sense) naturally appears. Indeed, this case boils
down to a special case of a game in which the action space of player 2 is
singleton, $\abs{AS_2}=1$, and we only look at the performance of different
actions, all of which may be static. Let us give two practical examples:

\begin{exa}[installing anti malware systems]\label{exa:anti-malware}
Installations of anti-virus software is a standard precaution, however,
given inconsistencies between reported performances and the diversity of
threats, it is often advisable to install several anti malware precautions.
No system ships with a guaranteed detection rate, and different systems may
be differently fast in identifying and blocking threats from the outside.
The decision must again be based on empirical data (reported recognition
rates) and just installing the full palette of available software will not
necessarily increase the protection. Thus, there is an eventual tradeoff
between security and cost, for which an optimum has to be found (see
Section \ref{sec:tradeoffs}).
\end{exa}

\begin{exa}[hiring security guards]\label{exa:guards}
As with example \ref{exa:usv} and as illustrated in Fig
\ref{fig:optimal-tradeoffs}, hiring more security guards will increase
likelihoods to catch an intruder, but also will increase costs. This is a
case where one performance indicator is random (the chances to keep
intruders outside), while the other is quite deterministic (the guard's
salaries). As before, the sought decision is a mere $\preceq$-optimal
selection among finitely many choices.
\end{exa}

Generally, in both examples, the problem is to find an optimal choice among
our options in $AS_1$, while the effects of a choice are not determined by
an adversarial move in the game, but rather up to our own assessment
(formally, we can simulate this by defining $AS_2$ to be singleton and
contain an abstract (not further specified) action). The technical issue
illustrated in both examples is the problem of $\preceq$-comparing
\emph{deterministic} to \emph{random} outcomes. This is indeed easily
possible and has been formally proven in \cite{Rass2015g} to work as
follows: let $a$ be a real value and let $B$ be a real-valued \ac{RV},
whose distribution is $f_B$ and is supported on $[1,b]$.

\begin{itemize}
  \item If $a\leq b$, then we $\preceq$-prefer $a$ (intuitively, this is
      because having a constant outcome $a$ is more secure to rely on than
      on a random effect $B$).
  \item If $a>b$, then we $\preceq$-prefer $B$ (intuitively, since it in
      any case admits less damage than $a$, whose damage is guaranteed).
  \item If $a=b$, then we prefer $B$ (intuitively, this is because $B$
      admits damages less than $a$, while the other action entails a
      guaranteed loss $\geq$ anything that can happen under $B$).
\end{itemize}

This procedure works on single criteria decisions only. If multiple
criteria are to be taken into account, then we must again resort to a
kernel density approximation to uniformly represent all values as \acp{RV}
(thus, adding a controllable/reasonable degree of uncertainty) to the
variable $a$, to be able to use the lexicographic comparisons on the convex
combination of utility functions as sketched in Section
\ref{sec:tradeoffs}.

\section{Example Application}\label{sec:concrete-example}
Continuing our running example from Section \ref{sec:running-example}, much
of the modeling has already been done along the vulnerability analysis
described in Section \ref{sec:running-example}. Indeed, the action sets
$AS_1$ and $AS_2$ are already available in Table \ref{tbl:as1} and Table
\ref{tbl:as2}, which makes the game a $4\times 8$-matrix over yet to be
specified outcomes. To simplify matters of demonstration here, let us use
only two strategies out of the sets $AS_1, AS_2$, leaving the full case of
our example or more extensive lists of attacks and countermeasures as an
obvious matter of scaling the matrix to a larger shape. The process of
finding the optimal risk mitigation strategy, however, remains unchanged
between a $2\times 2$- and an $n\times m$-game, so we will illustrate the
results on the smaller example without loss of generality. We stress that
even though $2\times 2$-games admit closed form solutions over $\R$, the
same formula in $\HR$ holds but cannot be practically evaluated in lack of
an explicit ultrafilter (alas, the existence of $\UF$ is assured only
non-constructively). Our chosen strategies are in abbreviated form shown in
Table \ref{tbl:attack-selection}.

\begin{table}[h!]
\caption{\textbf{Selected Strategies for the Example}}\label{tbl:attack-selection}
\begin{tabular}{|l|l|}
  \hline
  \textbf{Attack strategies} & \textbf{Defense actions} \\\hline
  $a_1$: buffer overflow exploits & $d_1$: patching \\
  $a_2$: remote access exploits & $d_2$: deactivation of services \\
  \hline
\end{tabular}
\end{table}

Note that these chosen attacks are indeed generic, as buffer overflows and
remote access is part of every attack listed in Table \ref{tbl:as2}. In
fact, this modeling appears practically reasonable, since the respective
countermeasure may ``break'' the attack vector at any stage (as long as it
breaks it at all). However, depending on how deeply the attacker has
penetrated the system already, the respective countermeasures may not be
effective any more. In light of this, let us follow the procedure outlined
in Section \ref{sec:apt-games} and assume that we have asked a group of six
domain experts about their opinions on the effectiveness of
countermeasures. With answers given in qualitative terms of saying that the
residual \emph{risk after mitigation} is either (H)igh, (M)edium or (L)ow,
assume that the answers were as listed in Table
\ref{tbl:example-assessments} (using the abbreviations from Table
\ref{tbl:attack-selection}). The concrete rating of the risk can be based
on the graph-theoretic distance between the attacker and its goal (as
discussed in Section \ref{sec:apt-games}). This view lets us convert the
qualitative ratings into numeric ranks, which are $L\mapsto 1, M\mapsto 2$
and $H\mapsto 3$, expressing that a ``low'' rating is based on the belief
that the attacker has penetrated only one access control so far (e.g.,
gained access to machine 0), while a ``high'' rating means that it is
already quite close to its goal (penetrated three access control systems
already up to only one exploit left towards full access on machine 2; cf.
Fig \ref{fig:apt-example-tva}).

\begin{table}[h!]
\caption{\textbf{Example Expert Assessments}}\label{tbl:example-assessments}
\begin{tabular}{|c|c|c|c|c|c|c|c|}
\hline
\multicolumn{2}{|c|}{\textbf{Scenario} $\downarrow$ / \textbf{Expert} $\to$} & 1 & 2 & 3 & 4 & 5 & 6\tabularnewline
\hline
\multirow{2}{*}{$a_{1}$} & $d_{1}$ & L & L & M & M & M & H\tabularnewline
\cline{2-8}
 & $d_{2}$ &  & H & H &  & H & M\tabularnewline
\hline
\multirow{2}{*}{$a_{2}$} & $d_{1}$ & H &  & H & M & L & H\tabularnewline
\cline{2-8}
 & $d_{2}$ & M & L & L & M & M & \tabularnewline
\hline
\end{tabular}
\end{table}

Note that Table \ref{tbl:example-assessments} shows empty cells, which
correspond to cases where an expert was silent (without explicit opinion)
on a specific scenario. Such missing information is only an organisational
inconvenience, yet causes no technical difficulties, as in that case, we
simply compile the loss distributions from the data that is available.

We define a \ac{KDE} using \eqref{eqn:kernel-approximation} for each
scenario, using the data from Table \ref{tbl:example-assessments}, such as
the loss distribution for the scenario $(a_1,d_2)$ would, for example, be
constructed from $N=4$ data points $(x_1,x_2,x_3,x_4) = (3,3,3,2)$.

We practically implemented this scheme in \texttt{R}
\cite{RDevelopmentCoreTeam2011}, using the available heuristics for
bandwidth selection that \texttt{R} provides (concretely Silverman's rule),
and obtained the distribution-valued matrix game shown in Fig
\ref{fig:example-game}, with the label ``\textsf{loss(i,j)}'' indicating
the scenario $(d_i, a_j)$ for $i,j\in\set{1,2}$, with meanings as told by
Table \ref{tbl:attack-selection}.

\begin{figure}
  \includegraphics[scale=0.6]{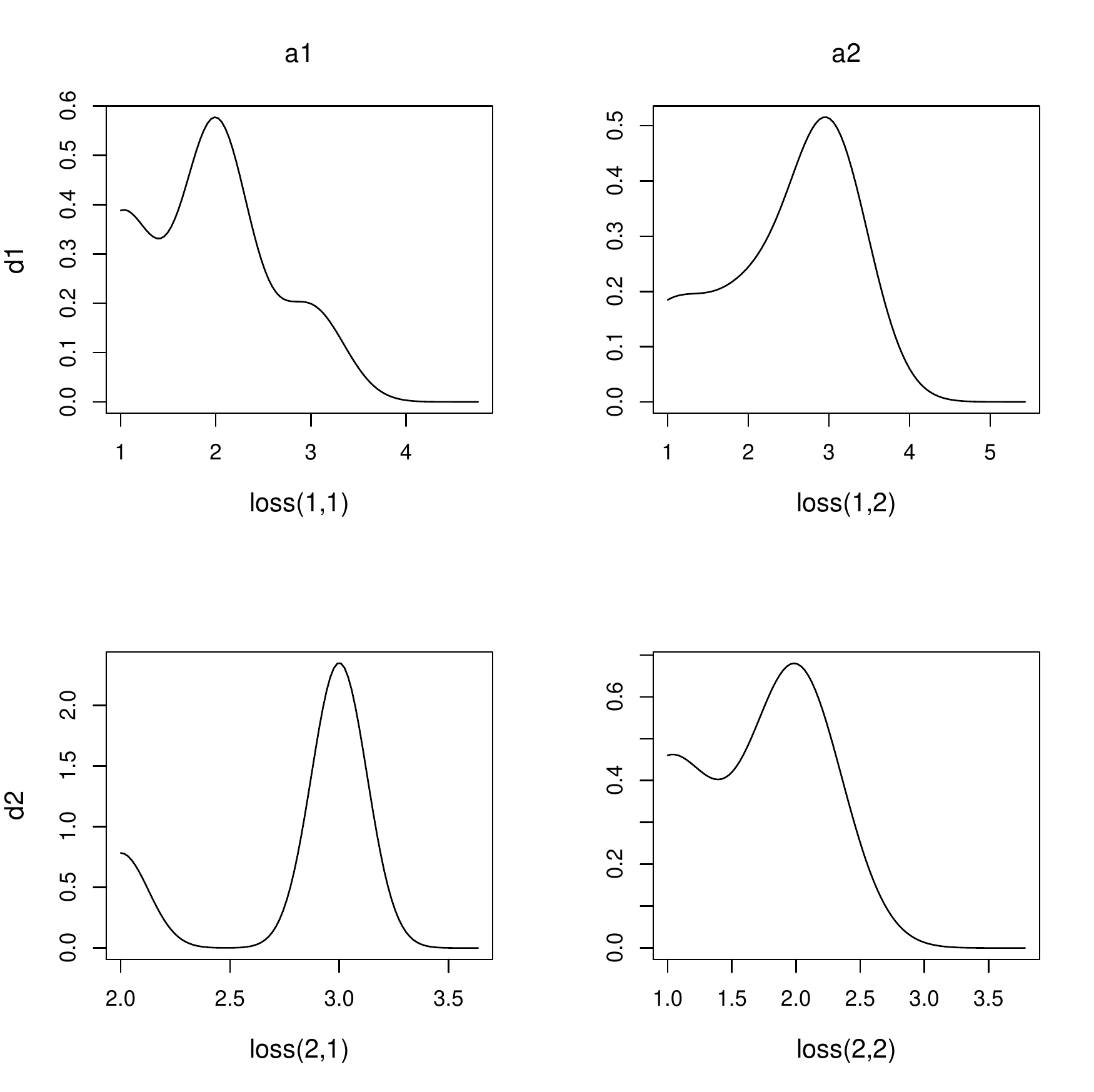}
  \caption{\textbf{\texttt{R}-plot of our example \ac{APT} matrix game}}\label{fig:example-game}
\end{figure}

To obtain a solution, we implemented a standard fictitious play algorithm
(see, e.g., \cite{Washburn2001}), with the only modification of minima and
maxima being selected w.r.t. the lexicographic ordering of the derivative
sequences \eqref{eqn:series-representation}. The derivatives themselves are
computed by evaluating \eqref{eqn:density-closed-form}. The used risk
acceptance threshold $a\geq 1$ in our implementation defaults to the
largest observation (data point) available. That is, we consider any event
with consequences above ``high damage'' as residual zero-day exploit risk,
which by construction of the $\preceq$-relation is minimized (cf. Theorem
\ref{thm:preference-meaning}).

Taking 1000 iterations of fictitious play and rounding the result to three
digits after the comma, we obtain the approximate lexicographic equilibrium $(\vec{\tilde
p}^*,\vec{\tilde q}^*)=((0.875, 0.125), (0.238, 0.762))$ along with the
resulting equilibrium loss distribution as shown in Fig
\ref{fig:equilibrium-loss}, and formally given as the derivative of the
distribution function $F(\vec {\tilde p}^*,\vec{\tilde q}^*)$ defined by
\eqref{eqn:utility-function}. Conceptually, this density is the same as the
(well known) saddle-point value of a regular game (it plays the same role
and the random loss corresponding to it satisfies the equilibrium condition
w.r.t. the $\preceq$-relation, more precisely, the lexicographic order on the probability masses distributed over the loss range).

Now, let us look at the practical meaning of the outcome of the
game-theoretic analysis:
\begin{itemize}
  \item The optimal way of mitigating the \ac{APT} as modeled by the game
      in Fig \ref{fig:example-game} is to do defense action $d_1$ with
      likelihood $0.875$ and defense action $d_2$ with likelihood
      $0.125$. That is, if actions are taken daily, then we would
      temporarily turn down a service (or enforce a disconnection) every
      8th day on average, while applying patches in the meantime,
      whenever they are available. Note that the randomness in the
      modeling \emph{accounts} for situations where a patch may not be
      available at all (in that case, the event of failure of the patch
      occurs, and high damage may be expected. This is, however, already
      captured by our modeling of a random outcome, which can have
      different effects, including a working and failing patch at
      different iterations of the gameplay).
  \item If the system administrator adheres to the (lexicographic) equilibrium behavior
      $\vec{\tilde p^*}=(0.875, 0.125)$ in choosing her/his actions, then
      Fig \ref{fig:equilibrium-loss} is by Lemma
      \ref{lem:zero-sum-optimality} a \emph{guarantee} concerning random
      damages, \emph{irrespectively} of how the attacker actually
      behaves. Indeed, a (non-unique) worst-case attack behavior is
      delivered as the second part of the equilibrium $\vec{\tilde
      q}^*=(0.238, 0.762)$, which tells that approximately every fourth
      attack ``should'' be a buffer overflow, while trying remote access
      in the remainder of the time. If our modeling was incorrect on
      assuming the attacker's behavior but correct on the possible
      actions (depending on how accurate the topological vulnerability
      scan was), then any behavior different from $\vec{\tilde q}^*$ will
      give us only less chances of high damage (as follows from the
      definition of an equilibrium; see \cite{Alpcan2010}).
  \item Conversely, the same also holds for the system administrator. The
      equilibrium condition tells that any attempt to do better by
      patching more often or deactivating services more or less
      frequently will render the loss bound $V=F(\vec{\tilde
      p}^*,\vec{\tilde q}^*)$ void (in the sense that losses are no
      longer optimally distributed), and may enable stronger worst-case
      attack scenarios (cf. Lemma \ref{lem:zero-sum-optimality}).
\end{itemize}

From the optimal distribution that is returned by the fictitious
play, we can easily compile other risk measures of interest like averages
(see \eqref{eqn:risk}) or similar. In particular, taking the expectation of
the distribution $F(\vec{\tilde p}^*,\vec{\tilde q}^*)$ directly returns a
quantitative risk estimate according to the common formula
\eqref{eqn:risk}. What is more, however, is our ability to obtain
additional information from the outcome, such as the variance as a measure
of ``stability'' of the risk estimate (quantified by the fluctuation around
the expected damage), or the danger of zero-day exploits, for which the
area within $[a,\infty)$ can be computed (shown in gray in Fig
\ref{fig:equilibrium-loss}) as an indication and decision support. The
considerable size of the gray area relative to the rest is due to the
expert's majority suspecting the risk to be high (cf. Table
\ref{tbl:example-assessments}). This puts a lot of mass on the right tail
of the constructed loss distributions, thus fattening the tails (induced by
the \ac{KDE} \eqref{eqn:kernel-approximation}) accordingly. A less
pessimistic expert assessment underlying the risk analysis would result in
a smaller zero-day likelihood, respectively.

These possibilities extend much beyond the usual capabilities of
(quantitative) risk management.

\section{Discussion}\label{sec:discussion}

\subsection{Actions in Continuous Time}\label{sec:continuous-time-actions}
Typically, cyber warefare and \acp{APT} are not games that take rounds, but
a process of actions in continuous time with no defined ``stages'' in the
game. More precisely, we have the situation of one player likely being
forced to take actions at discrete times, while the other player is free
act continuously over time. For example, the security officer may be unable
to become active at any (prescribed optimal) time of the day, since s/he
must adhere to organizational constraints of the daily business. On the
contrary, the adversary is only bound to the business organizational
matters as far as it concerns the mounting of an attack. In any case, the
attacker can act at any point in time (during night, during peak hours of
work, etc.), while the defender must use the proper time slots to become
active to minimize distortions in the actual enterprise business.

This is a remarkable qualitative difference to both, discrete and continuous
time game models, covering matrix games (discrete in time) and others (e.g.,
\cite{Dijk2013} as a continuous time model). Our specific \ac{APT} model is
discrete in time (as being a matrix game), but can account for randomness in
the outcome caused by a continuously acting opponent. This is a mere change
to the outcome (loss) distributions to cater for an action to be interrupted,
distorted or to cause more damage the longer it has happened in the past.

Theoretically, the embedding of the payoff distributions into $\HR$ equips
us with the full spectrum of mathematical tools as are known and applied to
construct continuous time games (this is a direction of future research
from this work). Practically, we believe the application of discrete time
games to more properly match the possibilities of a real-life security
officer, who may take actions only at particular times, i.e., outside
peak-hours of work, when devices are temporarily idle, or similar. Invoking
Lemma \ref{lem:zero-sum-optimality} here tells us that if the defense is
optimized for discrete time actions of the defender, a ``continuous time
attacker'' may nevertheless only deviate and as such cause less damage than
what the matrix game predicts.

\subsection{On Some Modeling Issues}\label{sec:faqs}
This section is similar to many well-known collections of ``frequently asked
questions'', and indeed can be taken as a guideline on how to overcome a set
of usual modeling obstacles when the model shall be applied practically:

\begin{quote}\textbf{How to deal with large attack graphs?}\end{quote}
As a matter of fact, attack graphs can get huge even for small
infrastructures. Enumerating all strategies by finding all paths can thus
become an infeasible task (as there is usually an exponential number of
them). Instead, we may either simulate attacks on the nodes directly, thus
defining the attack strategies not as paths in the attack graph but rather as
the set of possible exploits in our infrastructure, and ask for a qualitative
risk assessment based on the exploits only. Even if the infrastructure is
large, having the freedom to work with qualitative assessments in our
game-theoretic model eases matters of risk assessment essentially up to the
same complexity as a normal risk and vulnerability assessment would require.
That is, there is no conceptual obligation in our \ac{APT} game model
definition to work with paths in the attack graph (this is only \emph{one}
option among many), and the \ac{APT} game can be defined on aggregated parts
of an infrastructure, or using any other condensed or high-level view on the
system.

\begin{quote}\textbf{Where to get the losses/probabilities from?}\end{quote}
Probability -- in general -- is a notoriously opaque concept for many
practitioners, and specifying probability in our model, as for any
probabilistic model, is a crucial initial task. However, unlike many
other techniques, our requirements are different in an important way:
we do not ask an expert for a numeric estimate of a probability, but
instead it suffices to ask several experts for a qualitative rating of
likelihood concerning certain attacks. That is, the expert is not
challenged to tell a precise number to quantify how likely an attack
is, but can rather resort to saying that the success for an attack or
mitigation strategy may be ``low'', ``medium'' or ``high''. This is
even in alignment with recommendations of the German Federal Office for
Information Security (BSI), which explicitly warns about ``precise
probabilities'' which can misleadingly take estimates as being
objectively accurate. Asking several experts about their opinions and
re-normalizing the absolute frequencies into probabilities avoids the
common problems of calibrating probabilistic models and naturally
delivers the necessary loss distributions as we need them (see Fig
\ref{fig:workflow}).

Moreover, the criticism uttered against probabilities (the difficulty to
specify them and the illusion of accuracy created by them) applies much more
generally to many statistical models, but not as such in the model proposed
here.

An explicit alternative to surveys or mere expert opinions is the use
of \emph{simulation}. Models like \cite{Wellman2014} explicitly define
a continuous time simulation for a moving target defense that can be
adapted. Similar models from disease spreading analysis based on
percolation theory can also be used to probabilistically assess the
outcome of a malware infection (say via a bring-your-own-device
scenario) \cite{Koenig2016a} may also directly deliver the sought loss
distributions. The appeal of any such simulation is the automatism that
they provide to reduce the modeling labour.

\begin{quote} \textbf{What if a defense fails?}\end{quote}

Normal game theory assumes defenses to be effective in general, for
otherwise, an ineffective defense could be deleted from the game on
grounds of strategic dominance (cf. \cite{Gibbons1992}). In admitting
defense strategies to have random effects, the failure of a defense
is yet merely another form of failure of an action. Thus, upon proper
modeling of the chances for an action to fail, such events are
naturally covered by our random loss \acp{RV} admitting multiple
different outcomes, including a complete failure. For example, if
patches or updates are not available with known frequencies (say, if
the vendor has a ``patch day''), we can assign this likelihood as the
probability of high damage despite the action. Note, however, that
this is \emph{not} the freedom of choosing \emph{not} to do the
defense action at the prescribed time. Doing so would mean to deviate
from the equilibrium, which results in a worse protection. Adhering
to the equilibrium behavior strategy yet failing in the defense
action itself is, however, covered by the allowed likelihoods for an
action to fail at random.

\begin{quote}\textbf{How to handle one-shot situations?}\end{quote}

The existence of optimal decisions is often based on randomized
actions, which means that actions and the entire game can be
repeated. Practically, we may not be able to reset the
infrastructure to a defined initial condition after a damage
happened. That is, the game actually changes (at least temporarily)
upon past actions. Such situations are covered by the notion of
dynamic (stochastic) games, which allow future instances of the
game play to depend on past rounds. So, an action may be one-shot
and upon failure, may create a completely new situation. Assuming
that the entirety of possibilities admits a finite number of
game-theoretic descriptions, we technically have a stochastic game
in Shapley's sense \cite{Shapley1953a}. The usual notion of
equilibrium (optimal defense) in such competitions is, however,
more intricate and its existence is often tied to additional
assumptions or modifications to the game (e.g., by resorting to
dynamic games, or similar). The practical issue here is the
concrete choice of equilibrium outcomes (discounted, averaged,
etc.) to retain a practical meaning in the \ac{APT} context. We
avoid such difficulties here by allowing the outcome to be
different in each round and determined by past iterations of the
game, as long as the outcome remains \emph{identically distributed}
between rounds. If we think of the game structure itself following
a stochastic process, then the loss distributions constituting the
game structure may be taken as the stationary distribution of the
process, under any known condition of convergence (e.g., see
\cite{Levin2008}). We will leave this as a route for future work,
and close this discussion with the statement that one-shot
situations are conceptually equivalent to dealing with repeatable
situations, as we do not optimize the cumulative long-run average
(which would assume a repeated gameplay), but rather optimally
shape the distribution of the outcome for every repetition and thus
also for ``one-shots''.

\begin{quote}
\textbf{Is Knowledge about the Adversary's Incentives or Intentions
Required?}
\end{quote}

Adversary modeling is often perceived necessary or at least useful in
defending assets, especially in \ac{APT} scenarios. However, defenses
tailored to a specific guess about the adversary's intentions or incentives
may perform only suboptimal depending on the accuracy of the guess. Although
a game-theoretic model can be designed to take into account adversarial
payoffs if they are known, we do not actually require an accurate
understanding of the adversary's intentions or incentives.

It is important, however, to understand who the adversary is,
because this is what determines its action set $AS_2$. The more we
know about the attacker, the more accurate we can model its
actions, and thus reduce the possibilities for unexpected
incidents. Therefore, we must stress that Lemma
\ref{lem:zero-sum-optimality} only spares us to understand the
adversary's intentions, but \emph{we can in no case ignore} its
capabilities. Thus, we \emph{do require an adversary model} here,
yet being accurate only on the adversaries \emph{possible actions}
and on our own loss upon these.

\begin{quote}\textbf{How to set the risk acceptance threshold
$a$?}\end{quote} The risk acceptance threshold $a>1$ is here
required primarily for technical reasons, i.e., to assure the
boundedness of supports to ease deciding preferences among
actions. Thus, as long as any such value is being defined, the
theory and results remain intact. Physically, this parameter
corresponds to the maximal damage that we expect to occur ever,
or otherwise said, cases of damage that are covered by insurances
or considered so unlikely that the risk is simply taken.

Quite obviously, the concrete choice of the threshold $a$ has a
substantial influence on the decisions and computed equilibria.
Indeed, cutting off tails at different locations can even reverse
the preference under $\preceq$. Therefore, this value should be
chosen based on trial comparisons between actions to look for a
paradoxical/counter-intuitive outcome of $\preceq$ (see
\cite{Rass2015d} for an illustration), so that the value $a$ can
be set accordingly. Technically, it determines the region of loss
that we would relate to zero-day exploits (as was discussed in
the context of Fig \ref{fig:equilibrium-loss}). In any case, the
particular choice of $a$ is up to expertise, experience, and is
seemingly out of the scope of any default procedure to choose it.

Several rule of thumbs may be defined, such as choosing $a$ as a
maximum quantile over all loss distributions (similar to a
\ac{VaR} approach outlined in \cite{MacKenzie2014}), or directly
taking $a$ as the worst risk assessment made or possible. Our
implementation in \texttt{R} takes the maximum observation or
most pessimistic loss assessment as the cutoff point, assuming
that no more than the worst expected outcome may be expected.
However, if losses are quantified in monetary terms or general
business value, the lot of insurance coverage may determine the
acceptable risks that can be taken.

\section{Conclusion}\label{sec:conclusions}
Mitigating \acp{APT} on game theoretic grounds appears as a quite
natural model of the competition between a defender and an
attacker. The stealthiness of \ac{APT} adds an element of
uncertainty that original game theory covers with extended
notions like stochastic games or games with incomplete
information. Since these are conceptually more involved to
define, we propose staying with a simpler and easier to set up
model of matrix games. Deviating from classical game theory at
this point, we defined a concept that allows for ``direct use''
of vague and uncertain information that risk management normally
has to deal with. Specifically, lifting game theory from
real-valued payoffs to games whose outcomes are described by
entire probability distributions creates aspects of twofold
interest: practically, this model equips us with the full armoury
of game and decision theory to do risk management based on
uncertain and even qualitative information. Theoretically, the
so-generalized games come with substantially different properties
than their classical counterparts, such as the non-convergence of
fictitious play for a certain class of zero-sum games. The way in
which these issues are tackled is not tied to applications of
security and may thus be of independent interest in game theory.
Finally, in using matrix games with distributions as payoffs, we
tackle another aspect of \acp{APT}, which is the game being
discrete time for one player but continous time for the other
player. This aspect was hardly discussed in precursor work, and

For security, the game theoretic perspective lets us not only compute optimal
risk mitigation strategies almost directly starting from the available
information, but also elegantly saves us from some matters of adversary
modelling. Especially, we only need to know the attacker's possible actions,
but can work out a (multi-criteria) optimal defence in terms of our own risk
scale. This is particularly useful in the context of guarding against
\acp{APT}, since uncertainty is ``ubiquitious'' in the attacker's
capabilities, incentives, induced damages, etc. Having models that spare us
the need to model all these aspects, while dealing with uncertainty in the
way it comes (such as expert opinions on risk or expectations on zero-day
exploits) appears as a demanding issue. Our work is intended as a first step
into this direction.


%

\section*{Acknowledgment}
We thank the anonymous reviewers and the editorial board for
valuable suggestions that helped to improve the clarity and
content of this article, and also for drawing our attention to
interesting aspects to consider in future. This work was
supported by the European Commission's Project No. 608090, HyRiM
(Hybrid Risk Management for Utility Networks) under the 7th
Framework Programme (FP7-SEC-2013-1).


\section*{Appendix: Proof of the Approximation Theorem
\ref{thm:main}}

The following arguments are partly based on \cite{Rass2015g} and
close some open gaps in this preliminary draft presentation of
results. The proof of Theorem \ref{thm:main} rests on the
following lemma:
\begin{lem}\label{lem:derivative-preferences}
Let $f,g\in C^\infty([1,a])$ for a real value $a>1$ be probability density
functions. If
\[
    ((-1)^k\cdot f^{(k)}(a))_{k\in\N}<_{lex} ((-1)^k\cdot g^{(k)}(a))_{k\in\N},
\]
then $f\preceq g$.
\end{lem}
\begin{proof}
The proof repeats the arguments in \cite{Rass2015g}, and is
essentially based on Lemma \ref{lem:preference-decisions}. That
is, it suffices to determine which density is taking lower values
in a right neighborhood of the cutoff point $a$. To this end, let
us ``mirror'' the functions around the vertical line at $x=a$ and
look for which of $f(x), g(x)$ grows faster when $x$ becomes
larger than $a$, using an induction argument on the derivative
order $k$. Clearly, whichever function grows slower for $x\geq a$
in the mirrored view is the $\preceq$-preferable one by Lemma
\ref{lem:derivative-preferences}. Furthermore, we may assume
$a=0$ without loss of generality (as this is only a shift along
the horizontal line). For $k=0$, we have $f(0)<g(0)$ clearly
implying that $f\preceq g$, since the continuity implies that the
relation holds in an entire neighborhood $[0,\eps)$ for some
$\eps>0$. Thus, the induction start is accomplished.

For the induction step, assume that $f^{(i)}(0)=g^{(i)}(0)$ for all $i<k$,
$f^{(k)}(0)<g^{(k)}(0)$, and that there is some $\eps>0$ so that
$f^{(k)}(x)<g^{(k)}(x)$ is satisfied for all $0< x<\eps$. Take any such
$x$ and observe that
\begin{align*}
    0 &> \int_0^x\left(f^{(k)}(t)-g^{(k)}(t)\right)dt\\
    &=f^{(k-1)}(x)-f^{(k-1)}(0)-\left[g^{(k-1)}(x)-g^{(k-1)}(0)\right]\\
    &=f^{(k-1)}(x)-g^{(k-1)}(x),
\end{align*}
since $f^{(k-1)}(0)=g^{(k-1)}(0)$ by the induction hypothesis. Thus,
$f^{(k-1)}(x)<g^{(k-1)}(x)$, and we can repeat the argument until $k=0$ to
conclude that $f(x)<g(x)$ for all $x\in(0,\eps)$.

For returning to the original problem, we must only revert our so-far
mirrored view by considering $f(-x), g(-x)$ in the above argument. The
derivatives accordingly change into $\frac {d^k}{dx^k}f(-x)=(-1)^k
f^{(k)}(x)$, and the proof is complete.
\end{proof}

\begin{proof}[Proof of Theorem \ref{thm:main}]
This is an actually easy matter of collecting what we have
obtained in Section \ref{sec:derivative-decisions}. First, note
that w.l.o.g., we can represent the losses in $\vec A$ by
distribution functions $F_{ij}$ for $i=1,2,\ldots,n$ and
$j=1,2,\ldots m$. To ease notation in the following, let $i,j$ be
arbitrary, and abbreviate $F_{ij}$ as $F$. We will go through a
sequence of approximations of $F$, denoted as $\tilde{F}_1,
\tilde{F}_2, \tilde{F}_3$, and $\tilde{F}_4$, respectively, and
prove that the $L^1$-approximation error of the final
approximation $\tilde{F}_4$ can be made bounded by $\delta$ upon
proper constructions of the intermediate approximations.

To get started, let us get back to the mollifier approach
outlined in remark \ref{rem:continuous-losses}: we choose some
$h>0$ and define first approximation $\tilde{F}_1(h):=F\ast
K_h\in C^\infty$. Note that this already makes the sequence
representation \eqref{eqn:series-representation} well-defined.
Moreover, it is known that letting $h\to 0$, the sequence
$\tilde{F}_1(h)$ is $L^1$-convergent to $F$ by known
approximation theorems (e.g., on page 321 in
\cite{Koenigsberger2004}). That is, we can choose a sufficiently
small $h^*>0$ to have $\tilde{F}_1:=F\ast K_{h^*}$ satisfy
$\norm{\tilde{F}_1-F}_{L^1}<\delta/4$.

Note that since $K_{h^*}$ is supported on the entire real line, so is
$\tilde{F}_1$. To recover the required bounded support, we choose some value
$a>1$ and truncate the distribution $\tilde{F}_1$ outside the interval
$[1,a]$. Call the result $\tilde{F}_2$. Since $\tilde{F}_1$ is a probability
distribution, it satisfies $\lim_{x\to\infty}\tilde{F}_1(x)=1$, so that we
can choose $a$ sufficiently large to make the truncated distribution
$\tilde{F}_2$ satisfy $\norm{\tilde{F}_1-\tilde{F}_2}_{L^1}<\delta/4$ again.

Since $\tilde{F}_2\in C^\infty([1,a])$ by construction and the derivatives
are all continuous (and as such bounded on the compact interval $[1,a]$), we
can approximate $\tilde{F}_2$ at the point $x=a$ by a Taylor-polynomial
$\tilde{F}_3=\sum_{i=0}^K \frac{\tilde{F}_2^{(i)}(a)}{i!}\cdot (x-a)^i$. The
$i$-th derivative is analytically given by $\tilde{F}_2^{(i)}=F\ast
K_h^{(i)}$ and computed numerically by virtue of
\eqref{eqn:density-closed-form}. The accuracy of the Taylor polynomial
approximation is governed by choosing the order $K$ of the polynomial
sufficiently large. In our case, we take $K$ large enough to make the
approximation satisfy $\norm{\tilde{F}_2-\tilde{F}_3}_{L^1}<\delta/4$.

Finally, observe that the Taylor polynomial $\tilde{F}_3$ can be
represented by a finite sequence of its coefficients
$(a_0,a_1,\ldots,a_K)\in\R^K$ with $a_i := \tilde{F}_2^{(i)}(a)$.
To the end of recovering a finitely truncated representation as
in \eqref{eqn:series-representation}, define the coefficients
with alternating signs $b_i := (-1)^i a_i$ for $i=0,\ldots,K$.
Let us choose fixed integers $m,n$ and round all $b_i$ to $m$
places before and $n$ binary digits after the comma, padding with
leading and trailing zeroes. Call the resulting approximate
coefficients $\tilde b_i$, and define the respective binary
number $c := 0.\tilde b_1\|\tilde b_2\|\ldots\|\tilde b_K$ by
simply concatenating the bitstrings representing all $b_i$ in
ascending order of indices and omitting the decimal points.

Observe that the number $c$ is now a real value that encodes a
Taylor-polynomial $\tilde{F}_4$ with approximate coefficients $\tilde{a}_i$,
which differ from the coefficients $a_i$ in $\tilde{F}_3$ only by a rounding
error. Consequently, the maximal difference between $\tilde{F}_4$ and $\tilde
F_3$ comes to
\begin{align*}
 \max_{x\in[1,a]}\abs{\tilde F_3(x)-\tilde
F_4(x)} &\leq \max_{x\in[1,a]}\abs{\sum_{i=0}^K \frac{\abs{a_i - \tilde a_i}}{i!}\cdot (x-a)^i} \\
&\leq \max_{x\in[1,a]}\abs{\sum_{i=0}^\infty \frac{\eps_n}{i!}\cdot (x-a)^i} \leq \eps_n\cdot e^a,
\end{align*}
where $\eps_n$ is the maximal numeric roundoff error depending on
the number $n$ of digits after the comma. Since $a$ is a
constant, we can choose $n$ sufficiently large to make
$\norm{\tilde F_4-\tilde F_3}_{L^\infty([1,a])}$ sufficiently
small, and hence also cause $\norm{\tilde F_4-\tilde
F_3}_{L^1([1,a])}<\delta/4$ ultimately (this is a consequence of
using H\"older's inequality to show that convergence in the
function space $L^p$ implies convergence in $L^q$ for $q<p$ if
the underlying support is compact; see page 233 in
\cite{Elstrodt2002}).

Collecting the approximations obtained along, we end up finding that
$\norm{F-\tilde F_4}_{L^1}\leq\norm{F-\tilde F_1}_{L^1} + \norm{\tilde
F_1-\tilde F_2}_{L^1}+\norm{\tilde F_2-\tilde F_3}_{L^1}+\norm{\tilde
F_3-\tilde F_4}_{L^1} \leq 4\cdot\frac \delta 4=\delta$, as required.

Indeed, repeating these steps for every entry in the matrix $\vec
A=(F_{ij})\in\F^{n\times m}$, we end up with a matrix of
respective values $\vec B=(c_{ij})\in\R^{ij}$ representing finite
sequence approximations of $F_{ij}$. Moreover, observe that the
construction of the matrix $\vec B$ is such that the numeric
order between two entries is exactly the lexicographic order on
the sequence of rounded coefficients (this is due to the
concatenation and the fact that all numbers $b_i$ are represented
use the same number of digits before and after the comma). Hence,
Lemma \ref{lem:derivative-preferences} tells that the order of
choices made in the game $\vec B$ equals the $\preceq$-order of
choices that would be made in the game $\vec A$. Consequently, an
equilibrium in $\vec B$ is a lexicographic equilibrium in $\vec A$ since
$\leq$ on the so-obtained $c$-values in $\vec B$ equals $\preceq$
on the original loss distributions in $\vec A$.

Since $\vec B$ is a regular matrix game, we can invoke fictitious play (or
linear optimization) to compute an approximate (or even accurate) standard Nash equilibrium
$(\vec p^*,\vec q^*)$ in $\vec B$ at any desired precision $\eps$. By 
construction of $\vec B$ and \eqref{eqn:utility-function}, 
$(\vec p^*,\vec q^*)$ will approximate a lexicographic equilibrium payoff in $\vec A$.
This completes the
proof.
\end{proof}

As a final remark, note that the representation of the
Taylor-polynomial within the real number $c$ is compatible with
the multi-criteria optimization as outlined in Section
\ref{sec:tradeoffs}. To see this, observe that the convex
combination \eqref{eqn:convex-combination} is a linear operation
whose result remains within the same bound as the inputs. Thus,
there will be no ``overflow carry'' from one coefficient to the
next in the representation $c$ of the Taylor-polynomial.

\begin{acronym}
\acro{FTP}{file transfer protocol}%
\acro{RSH}{remote shell}%
\acro{SSH}{secure shell}%
\acro{KDE}{kernel density estimator}%
\acro{APT}{advanced persistent threat}%
\acro{CVSS}{common vulnerability scoring system}%
\acro{EFG}{extensive form game}%
\acro{RV}{random variable}%
\acro{TVA}{topological vulnerability analysis}%
\acro{VaR}{value at risk}%
\acro{IPR}{intellectual property rights}%
\end{acronym}

\end{document}